\documentclass[12pt, onecolumn]{IEEEtran}

\usepackage{cite}
\usepackage{pslatex}
\usepackage{moreverb}
\usepackage{epsfig}
\usepackage{amsmath}
\usepackage{amssymb}
\usepackage{bm}
\usepackage{booktabs}
\usepackage{array}
\usepackage{multirow}
\usepackage{threeparttable}
\usepackage{url}
\usepackage{dsfont}
\usepackage{subfigure}
\usepackage{midfloat}
\usepackage{color}
\usepackage{graphicx,dblfloatfix}

\newtheorem{lemma}{Lemma}
\newtheorem{thm}{Theorem}
\newtheorem{definition}{Definition}
\newtheorem{corollary}{Corollary}

\def\bc{{\mathbf c}}

\def\bw{{\mathbf w}}

\def\by{{\mathbf y}}

\def\bH{{\mathbf H}}

\def\b0{{\mathbf 0}}

\DeclareMathAlphabet{\mathpzc}{OT1}{pzc}{m}{it}

\renewcommand{\baselinestretch}{1.785}
\begin{document}
\vspace{-6mm}
\title{Multi-user lattice coding for the multiple-access relay channel}

\author{
    \IEEEauthorblockN{Chung-Pi Lee, Shih-Chun Lin, Hsuan-Jung Su and H. Vincent Poor}\\
\thanks{C.-P. Lee and H.-J. Su are with Department of Electrical Engineering and Graduate Institute of Communication Engineering,
National Taiwan University, Taipei, Taiwan, 10617. S.-C. Lin is with
the Department of Electronic and Computer Engineering, National
Taiwan University of Science and Technology, Taipei, Taiwan, 10607.
H. V. Poor is with the Department of Electrical Engineering,
Princeton University, Princeton, NJ, 08544, USA. Emails:
\{D96942016@ntu.edu.tw, sclin@mail.ntust.edu.tw,
hjsu@cc.ee.ntu.edu.tw, poor@princeton.edu\}. }}

 \maketitle \thispagestyle{empty}
 \vspace{-20mm}
\begin{abstract}
 \vspace{-2mm}
This paper considers the multi-antenna multiple access relay channel
(MARC), in which multiple users transmit messages to a common
destination with the assistance of a relay. In a variety of MARC
settings, the dynamic decode and forward (DDF) protocol is very
useful due to its outstanding rate performance. However, the lack of
good structured codebooks so far hinders practical applications of
DDF for MARC. In this work, two classes of structured MARC codes are
proposed: 1) one-to-one relay-mapper aided multiuser lattice coding
(O-MLC), and 2) modulo-sum relay-mapper aided multiuser lattice
coding (MS-MLC). The former enjoys better rate performance,
 while the latter provides more flexibility to
tradeoff  between the complexity of the relay mapper and the rate
performance. It is shown that, in order to approach the rate
performance achievable by  an unstructured codebook with
maximum-likelihood decoding,  it is crucial to use a new $K$-stage
coset decoder for structured O-MLC, instead of the one-stage decoder
proposed  in previous works. However, if O-MLC is decoded with the
one-stage decoder only, it can still achieve the optimal DDF
diversity-multiplexing gain tradeoff in the high signal-to-noise
ratio regime. As for MS-MLC, its rate performance can approach that
of the O-MLC by increasing the complexity of the modulo-sum
relay-mapper. Finally, for practical implementations of both O-MLC
and MS-MLC,  practical short length lattice codes with linear
mappers are designed, which facilitate efficient lattice decoding.
Simulation results show that the proposed coding schemes outperform
existing schemes in terms of outage probabilities in a variety of
channel settings.
\end{abstract}


\vspace{-5mm}
\section{Introduction}
\vspace{-2mm} In recent years, cooperative communication has drawn a
significant amount of interest as a means of providing spatial
diversity  when time, frequency or antenna diversities are
unavailable due to delay, bandwidth or terminal size constraints,
respectively. Cooperative communication techniques for single-source
networks have been extensively studied in terms of rate, outage
probability or diversity-multiplexing tradeoff (DMT) perspectives
\cite{funo}\cite{cooperaive_outage}\cite{on_the_achievableDMT}.
However, practical communication networks usually involve more than
one source (user), leading to the study of   the multiple-access
channel (MAC).
In this paper, we consider an important multi-user cooperative
communication channel, that is, the multi-antenna multiple-access
relay channel (MARC). The MARC is a MAC with an additional shared
half-duplex relay\cite{MARCfirst}.
 It has been shown
that the MARC provides a much larger achievable rate region
\cite{MARCfirst} and diversity gain per user \cite{case_MAC},
compared to those of the MAC.
Also, since a single relay is shared by multiple users in the MARC,
the extra cost of adding such a relay is acceptable. However, the
code design for the MARC needs to jointly consider the codebooks of
the multiple users and the relay
\cite{MARCfirst}\cite{Capcity_MARC2}\cite{cooperaive_capacity}, and
is thus not a trivial extension of those for the single-user relay
channel or the multiple access channel.

The achievable rate region of the MARC has been characterized in
\cite{MARCfirst} \cite{Capcity_MARC2} and
\cite{cooperaive_capacity}. The decode and forward protocol, which
is a special case of the dynamic decode and forward (DDF) protocol
\cite{optimality_ARQ}, was shown to achieve the capacity region of
the  MARC when the source-relay link is good enough{
\cite{cooperaive_capacity}},  thus having a larger achievable rate
region than those of the multiple-access amplify and forward (MAF)
\cite{case_MAC} and
 compress and forward (CF) protocols \cite{cooperative_wirelss}. However, the
capacity region of the general MARC remains unknown. The DMT for the
MARC with single antenna nodes was studied in \cite{case_MAC}
\cite{optimality_ARQ} and \cite{cooperative_wirelss}. Although the
MAF and CF are both DMT optimal in the high multiplexing gain regime
\cite{case_MAC}\cite{cooperative_wirelss},  compared with the DDF
strategy, they both achieve lower diversity gains in the low to
medium multiplexing gain regimes
\cite{case_MAC}\cite{cooperative_wirelss}.
Moreover, in \cite{case_MAC}, simulation results show that the DDF
protocol yields a better outage probability than that of MAF and CF
over a large range of  signal-to-noise ratio (SNR), even at the high
multiplexing gain regime. Thus we focus on the DDF in this paper due
to its good performance in a variety of operation settings.

However, previous results in
\cite{MARCfirst,Capcity_MARC2,cooperaive_capacity,case_MAC,optimality_ARQ,cooperative_wirelss}
are based on \textit{unstructured} random codebooks and maximum
likelihood (ML) decoders, and are very difficult to implement in
practice. In this paper, we propose \textit{structured} multiuser
lattice coding aided by a relay mapper for the MARC under the DDF
protocol, in which each node in the MARC has multiple antennas. To
simplify the joint codebook design problem for the multiple users
and the relay, we introduce a relay mapper which selects the
codeword to be transmitted at the relay to aid the users'
transmissions. The relay mapper is a key new ingredient for our
coding design, which can also help  implement the unstructured
codebooks in \cite{MARCfirst,Capcity_MARC2,cooperaive_capacity} and
\cite{optimality_ARQ} in practice, and does not appear in
\cite{MARCfirst,Capcity_MARC2,cooperaive_capacity,case_MAC,optimality_ARQ,cooperative_wirelss}.
However, the introduction of the relay mapper makes the decoding
much more difficult than that for the MAC \cite{latticemac}. We will
see that the one-stage  coset decoding  proposed in
\cite{latticemac} fails to achieve the rate performance of the
unstructured codebook with the ML decoding demonstrated in
\cite{cooperaive_capacity}. Instead, we propose a new $K$-stage
coset decoder that achieves the rate performance in
\cite{cooperaive_capacity} by successive cancellation on the
multiuser decoding tree. Two classes of relay mapper aided multiuser
lattice coding are proposed: 1) one-to-one relay mapper aided
multiuser lattice coding (O-MLC), and 2) modulo-sum relay mapper
aided multiuser lattice coding (MS-MLC). The first  enjoys better
rate performance while the second  provides more flexibility to
tradeoff between the complexity of the relay mapper and the rate
performance.  With the $K$-stage coset decoder, the structured O-MLC
can achieve the rate performance obtained by the unstructured
codebook in \cite{cooperaive_capacity}. If only one-stage coset
decoding is used, we also show that O-MLC is DMT optimal for the
DDF, and has better DMT than that in \cite{case_MAC} and
\cite{cooperative_wirelss} for the low to medium multiplexing gain
regime. As for MS-MLC, when the codomain size of the modulo-sum
relay mapper becomes larger, the error performance of  MS-MLC
approaches that of  O-MLC. Moreover, our decoder is no longer a
simple lattice decoder  as that of \cite{latticemac}, since the
lattice structure for decoding may be destroyed by the relay mapper.
Further, a naive application of the theoretical error analysis in
\cite{latticemac} suffers from significant losses in prediction of
the achievable rates of proposed coding. { We overcome this problem
by introducing a new technique for bounding the error probability
over the random relay-mapper codebook ensemble.} Finally, to
implement our theoretical results, we construct  practical lattice
codebooks with \emph{linear} mappings for both O-MLC and MS-MLC,
which enable the decoder to use the efficient lattice decoding
algorithms in \cite{aunified} and \cite{onmax}.

Compared with codes appearing in previous  works
\cite{MARCfirst,Capcity_MARC2,cooperaive_capacity,optimality_ARQ,cooperative_wirelss}
which are difficult to  implement, our structured MARC coding can be
implemented in practice as we  will see below. Some practical MARC
code designs were proposed in \cite{MARCnetowrk2} and
\cite{complexnetwork}, but these studies lack  theoretical
performance analysis. In \cite{MARCnetowrk2} and
\cite{complexnetwork}, an orthogonal protocol was used in which
users and the relay must transmitted in different time slots to
avoid interference, while our scheme allows them to transmit
simultaneously. Moreover, in \cite{complexnetwork}, instead of joint
code design,  the relay's transmitted symbol is formed  from the
users' symbols with a simple transformation.
 Due to the above reasons, there
are significant losses in the achievable rates and DMTs for the
methods in \cite{MARCnetowrk2} and \cite{complexnetwork}, compared
with our schemes. In simulations, we show that our proposed lattice
coding schemes also outperform the schemes in \cite{case_MAC}
\cite{cooperative_wirelss} \cite{MARCnetowrk2} and
\cite{complexnetwork} in terms of outage probabilities.

 The rest of this paper is
organized as follows. Section \ref{sec_systemmodel} introduces the
system model and some frequently used notation is summarized in
Table \ref{notations_table}. In Section \ref{proposed mapper}, O-MLC
and  MS-MLC are introduced. In Section \ref{sec.Error_Analysis}, we
establish the achievable rate region for both O-MLC and MS-MLC and
show that  O-MLC is DMT optimal. In Section \ref{Simulation},
simulation results are presented, and Section \ref{sec_conclusion}
concludes the paper.

\vspace{-4mm}
\section{System Model} \label{sec_systemmodel}
\vspace{-2mm} We consider the $K$-user multiple-antenna MARC as
shown in Fig. \ref{fig.1},
\begin{figure}[t] 
\centering
\includegraphics[scale=0.8]{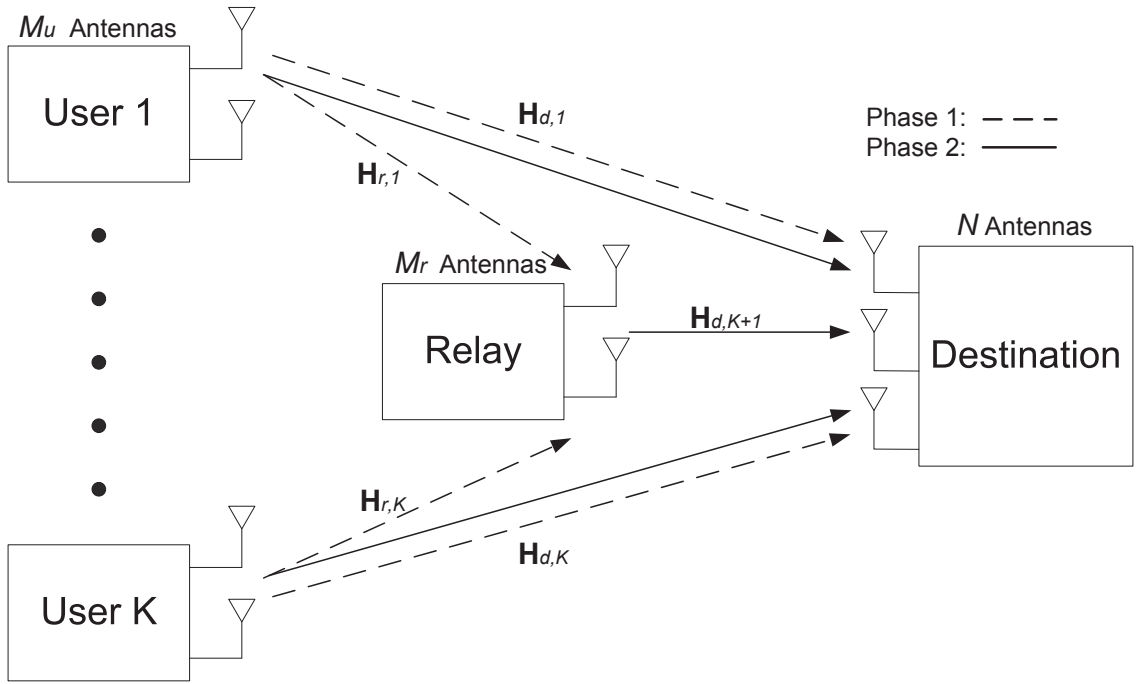}
\caption{ Dynamic decode and forward (DDF) for the $K$-user multiple-antenna  multiple-access relay channel  (MARC), where Phase $1$ is the relay's listening phase while Phase $2$ is the relay's transmitting phase.} 
\label{fig.1}
\end{figure}
in which a relay node is assigned to
assist the multiple-access users in transmitting data to a common
destination. Each user and the relay  is equipped  with $ M_{u}$ and
$M_{r}$  antennas, respectively, and the destination has $N$
antennas. In the DDF for MARC, each codeword spans $L$ slots each
consisting of $T$ vector symbols, and the block of $LT$ vector
symbols  is split into two phases due to the half-duplex constraint
at the relay node (i.e., it cannot transmit and receive
simultaneously).  In  Phase $1$, the relay receives the signals from
the users, then it tries to decode the users' messages until the
decision time $\ell_1T$. Following \cite{optimality_ARQ}, $\ell_1T$
is chosen to be the earliest time index such that after $\ell_1T$
symbols, the relay can decode the users' messages without error. If
there is no such $\ell_{1} \in \{1,...,L-1\}$, the relay remains
silent. Let the $M_{r}\times M_{u}$, $N \times M_{u}$ channel
matrices from user $i$ to the relay and the destination be
$\mathbf{H}_{r,i}$ and $\mathbf{H}_{d,i}$, respectively, which are
perfectly known at the corresponding receivers.  For  Phase $1$, the
received $M_{r} \times 1$ vector of symbols at the relay
is\footnote{\underline{Notation}: Let $A$ be a set, then $A^{*}=
A\setminus \{\mathbf{0}\}$. $A^{c}$ denotes the complement of $A$,
and $|A|$ denotes the cardinality of  $A$. For a matrix
$\mathbf{M}$, $\mathbf{M}^{H}$ is the
 conjugate transpose and $|\mathbf{M}|$ is the determinant. We use $\log(\cdot)$ for the logarithm with base 2, and
$\times$ for the direct product. An $n$-dimensional real lattice
${\Lambda}$ is a discrete additive subgroup of $\mathbb{R}^{n}$. The
lattice quantization function is defined as $
Q_{\Lambda}(\mathbf{y})\triangleq \arg \min_{\lambda \in
\Lambda}|\mathbf{y}-\mathbf{\lambda}| $ for
$\mathbf{y}\in\mathbb{R}^{n}$, and the modulo-lattice operation $
\bar{\by}=\mathbf{y}\mod \Lambda \triangleq
\mathbf{y}-Q_{\Lambda}(\mathbf{y})$ \cite{nested_lattice}.
 The second-order moment of $\Lambda$ is
defined as $
\sigma^{2}(\Lambda)\triangleq\frac{1}{nV_{f}(\Lambda)}\int_{\mathcal{V}_{\Lambda}}\mathbf{x}^{2}d\mathbf{x}
$, where $\mathcal{V}_{\Lambda}$ and $V_{f}(\Lambda)$ are given in
(T1.2) and (T1.3) in Table \ref{notations_table}, respectively. Some
other frequently used notation is also summarized in Table
\ref{notations_table}. }
\begin{equation}
\mathbf{y}_{r,l}=\sqrt{\frac{\rho_r}{M_{u}}}\sum^{K}_{i=1}\mathbf{H}_{r,i}\mathbf{x}_{i,l}+\mathbf{n}_{l},
\;\;\;\;\; l=1,2,...,\ell_1T \label{eqd.2}
\end{equation}
where  $\rho_r$ is the received SNR at the relay, $\mathbf{x}_{i,l}$
is the $M_{u} \times 1$ vector signal transmitted by user $i$ at
time index $l$, and the noise at the relay $\mathbf{n}_{l}\sim
\mathcal{CN}(\mathbf{0},\mathbf{I}_{M_{r}})$ is a Gaussian vector
with independent and identically distributed (i.i.d.) entries.
Similar to \eqref{eqd.2}, the received vector symbols at the
destination in Phase $1$ is
\begin{equation}
\mathbf{y}_{d_1,l}=\sqrt{\frac{\rho_d}{M_{u}}}\sum^{K}_{i=1}\mathbf{H}_{d,i}\mathbf{x}_{i,l}+\mathbf{v}_{l},
\;\;\;\;\; l=1,2,...,\ell_1T \label{eqd.3}
\end{equation}
where $\rho_d$ is the received SNR  and $\mathbf{v}_{l}\sim
\mathcal{CN}(\mathbf{0},\mathbf{I}_{N})$ is the noise vector at the
destination. In  Phase $2$ of DDF, based on the decoded messages
obtained at the decision time $\ell_{1}T$, the relay transmits the
corresponding coded vector symbols to the destination. The signal
received by the destination is then
\begin{equation}
\mathbf{y}_{d_2,l}=\sqrt{\frac{\rho_d}{M_{u}}}\sum^{K}_{i=1}\mathbf{H}_{d,i}\mathbf{x}_{i,l}+\sqrt{\frac{\rho_d}{M_{r}}}\mathbf{H}_{d,K+1}\mathbf{x}_{K+1,l}+\mathbf{v}_{l},
\;\;\;\;\; l=\ell_1T+1,\ell_1T+2,...,LT \label{eqd.4}
\end{equation}
where $\mathbf{x}_{K+1,l}$ denotes for the signal transmitted by the
relay and $\mathbf{H}_{d,K+1}$ is the channel matrix from the relay
to destination. As for the (normalized) MARC input power constraint,
it is imposed on each user and the relay as \vspace{-4mm}
\begin{equation}
E\left[\frac{1}{LT}\sum_{l=1}^{LT}|\mathbf{x}_{i,l}|^{2}\right]\leq
M_{u},
\;\;\;\;\;E\left[\frac{1}{LT}\sum_{l=1}^{LT}|\mathbf{x}_{K+1,l}|^{2}\right]\leq
M_{r}, \;\;\;\; i=1,...,K \label{eqd.5a}
\end{equation}
where the expectation $E[\;]$ is taken over all codewords in the
codebook.

To simplify the presentation for the proposed lattice coding scheme,
it is useful to transform our received signal model (\ref{eqd.2}),
(\ref{eqd.3}) and (\ref{eqd.4}) into the equivalent real channel
model form as in (\ref{eq_y_relay}) and (\ref{eqYdst}), for the
relay and the destination, respectively, \vspace{-4mm}
\begin{align}
\mathbf{y}_{relay}&=\mathbf{H}_{relay}\mathbf{x}_{relay}+\mathbf{n}_{relay}
\label{eq_y_relay}\\
 \mathbf{y}_{dst}&=\mathbf{H}_{dst}\mathbf{x}_{dst}+
\mathbf{n}_{dst}. \label{eqYdst}
\end{align}
The equivalent channel for the destination \eqref{eqYdst} is
formed by concatenating the received signal (\ref{eqd.3}) in Phase
$1$ and (\ref{eqd.4}) in Phase $2$, and the $2(KM_{u}+M_{r})LT
\times 1$ super signal vector $\mathbf{x}_{dst}$ in \eqref{eqYdst}
is
\begin{equation} \label{eq_x_dst}
\mathbf{x}_{dst}\triangleq\left[\mathbf{x}_{1}^{T},...,\mathbf{x}_{K+1}^{T}\right]^{T},
\end{equation}
where $\mathbf{x}_{i}=\left[\{\mathbf{x}^R_{i,1}\}^{T}, ...,
\{\mathbf{x}^R_{i,LT}\}^{T}\right]^{T}$ with
$\mathbf{x}^R_{i,l}=\left[\text{Re}\{\mathbf{x}_{i,l}\}^{T},\text{Im}\{\mathbf{x}_{i,l}\}^{T}\right]^{T}$;
while the $2NLT \times 1$ super received signal and noise at the
destination $\mathbf{y}_{dst}$ and $\mathbf{n}_{dst}$ in
\eqref{eqYdst} are similarly defined respectively. The $2NLT \times
2(KM_{u}+M_{r})LT$ super-channel matrix $\mathbf{H}_{dst}$ in
\eqref{eqYdst} is $\mathbf{H}_{dst}\triangleq
\left[\mathbf{H}^d_{1},...,\mathbf{H}^d_{K+1}\right]$, where the
$2NLT \times 2M_{u}LT$ equivalent channel matrix $\mathbf{H}^d_{i}$
for user $i$ comes from \eqref{eqd.3} as
\begin{equation}
\mathbf{H}^d_{i}\triangleq\sqrt{\frac{\rho_d}{M_{u}}}\mathbf{I}_{LT}\otimes
\begin{pmatrix}
\text{Re}\{\mathbf{H}_{d,i}\} & -\text{Im}\{\mathbf{H}_{d,i}\} \\
\text{Im}\{\mathbf{H}_{d,i}\} & \text{Re}\{\mathbf{H}_{d,i}\}
\end{pmatrix} \label{eqd.7}
\end{equation}
where $\otimes$ denotes the Kronecker product and $i=1,...,K$,
while the equivalent channel matrix $\mathbf{H}_{K+1}$ for the
relay comes from \eqref{eqd.4} as
\begin{equation}
\mathbf{H}^{d}_{K+1}\triangleq diag\left(
\mathbf{I}_{\ell_1T}\otimes \mathbf{0}_{2N\times2M_{r}},
\sqrt{\frac{\rho_d}{M_{r}}}\mathbf{I}_{(L-\ell_1)T}\otimes
\begin{pmatrix}
\text{Re}\{\mathbf{H}_{d,K+1}\} & -\text{Im}\{\mathbf{H}_{d,K+1}\} \\
\text{Im}\{\mathbf{H}_{d,K+1}\} & \text{Re}\{\mathbf{H}_{d,K+1}\}
\end{pmatrix}\right),
\label{eq.relaychannel}
\end{equation}
if $1\leq\ell_{1}\leq L-1$, where the first
$2N\ell_{1}T\times2M_{r}\ell_{1}T$ is a zero matrix because the
relay is  listening in Phase $1$ (if $\ell_{1}=L,$
$\mathbf{H}^d_{K+1}\triangleq \mathbf{0}_{2NLT\times2M_{r}LT}$ since
the relay is silent). As for the equivalent channel for the relay
(\ref{eq_y_relay}), it can be similarly obtained from \eqref{eqd.2}
as above, with the dimensions of $\bH_{relay}$ being $2M_{r}LT\times
2KM_{u}LT$. We consider two kinds of channel settings, the fixed
channel and the slow fading channel. In the fixed channel setting,
the channels are deterministic and we use the achievable rate as a
performance metric. For the slow fading channel, $
 \mathbf{H}_{dst}$ and $\mathbf{H}_{relay}$ are random but remain
constant over the whole code block. Since the MARC cannot support
any non-zero rate pairs with vanishing error probabilities now, we
use the DMT or the outage probabilities as performance metrics. The
entries of the channel matrices are assumed to be  i.i.d.
$\mathcal{CN}(0,1)$ when they are slow faded; i.e., we assume
Rayleigh fading in this case.



\vspace{-4mm}
\section{Proposed Relay-Mapper Aided Multiuser Lattice coding schemes}\label{proposed mapper}
\vspace{-2mm}
 In this section, we specify the proposed multiuser lattice coding
schemes for the MARC, i.e., O-MLC and  MS-MLC. Each of O-MLC and
MS-MLC consists of three building blocks: \textsl{1)}  the relay
mapper which decides which codeword to be transmitted at the relay,
\textsl{ 2)}  Loeliger-type nested lattices for the users' and the
relay's codebooks and \textsl{3)} a $K$-stage coset decoder, which
generalizes the one-stage decoder of \cite{latticemac}. We first
briefly introduce the adopted lattice codebooks. Tailored for them,
the relay mappers,
  the one-to-one mapper $\psi^{one}$ and the
modulo-sum mapper $\psi^{mod}$, for O-MLC and MS-MLC, respectively
are  shown in  Section \ref{sec_mapper}. Then the whole
encoding/decoding blocks  are introduced in Section
\ref{sec.lattice_coding}.

\vspace{-6mm}
\subsection{Loeliger-type Nested Lattice Codebooks} \label{sec_latticecodebook}
\vspace{-2mm}
 In our code construction,  codebooks of the $i$-th user ($1 \leq i \leq K$) and the relay ($i=K+1$) are generated from  Loeliger-type
 nested lattices.  To be specific, we introduce the following definitions.
\begin{definition} [Self-similar nested lattice code] \label{Def_nested}
For user $i$, let $\Lambda_{C_{i}}$ be a $2M_uLT$-dimensional coding
lattice and  ${\Lambda_{S_{i}}} \subset \Lambda_{C_{i}}$ be the
shaping lattice. The nested lattice codebook is defined as
$C_i^{nest}\triangleq\{\bar{\mathbf{c}}_{i}:
\bar{\mathbf{c}}_{i}=\mathbf{c}_{i}\mod \Lambda_{S_i},
\mathbf{c}_{i}\in \Lambda_{C_i}\}$, where $\bar{\mathbf{c}}_{i}$ are
the coset leaders \cite{nested_lattice} of the partition
$\Lambda_{C_i}/\Lambda_{S_i}$ (the set of cosets of $\Lambda_{S_i}$
relative to $\Lambda_{C_i}$). The codebook size is
$|C_i^{nest}|=2^{{R_{i}}{LT}}$, where the code rate is $R_i$ bits
per channel use (BPCU). When
$\Lambda_{S_i}=(2^{R_i/2M_u})\Lambda_{C_i}$ where $(2^{R_i/2M_u})
\in\mathbb{N}$ is the nesting ratio, the nested lattice code
$C_i^{nest}$ is called a self-similar nested code.\footnote{ Our
results can be easily generalized to the case in which good (but
maybe not self-similar) nested codes as in \cite{latc} are used.}
\end{definition}
For a Loeliger-type nested-lattice ensemble, the coding lattice
$\Lambda_{C_i}$ for user $i$ is randomly chosen from the Loeliger
lattices ensemble which is generated from linear codes $C^{Lo}_i$
\cite{avgboundlattice}. The detailed definition is given in
Definition \ref{Def_logier} in the Appendix \ref{MHthm1proof}-(I).

The codebook for the relay is generated similarly as above but with
dimension $2M_rLT$.

\vspace{-6mm}
\subsection{Proposed Relay Mappers} \label{sec_mapper}
\vspace{-2mm} The relay mapper $\psi$ is used to select the
codeword (coset leader) $\bar{\bc}_{K+1}$
 to be transmitted from the relay (transmitter $K+1$)
according to the codewords (coset leaders) $\bar{\bc}_i$,
$i=1,..,K$, of the $K$ users. In other words, by concatenating the
total $K+1$ codewords as a super one $\bar{\mathbf{c}}=
[(\bar{\bc}_1^T, \ldots,\bar{\bc}_K^T),
\bar{\bc}_{K+1}^T]^T=[\bar{\bc}_u^T,\bar{\bc}_r^T]^T$((T1.5) in
Table \ref{notations_table}),
\renewcommand{\baselinestretch}{1}
\begin{table}
\begin{threeparttable}
 \caption{List of Frequently used notation}
 \label{notations_table}\centering
\begin{tabular}[p{1cm}]{lp{7cm}p{5cm}}
\toprule[0.8pt]
Notation & Definition & Description \\
\hline
(T1.1)\; $\mathbb{Z}_{p}^{n}$ &  $n$-dimensional finite field over
$\mathbb{Z}_{p}=\{0,1,...,p-1\}$, where $p$ is a prime & Prime $p$ finite field\\
(T1.2)\; $\mathcal{V}_{\Lambda}$& The set of $\mathbf{v}\in
\mathbb{R}^{n}$ closer  to $\mathbf{0}$ than to any other  $
{\lambda} \in {\Lambda}$, for a lattice ${\Lambda}$ & Voronoi Region\\
(T1.3)\; $V_{f}(\Lambda)$ & Volume of Voronoi region $\mathcal{V}_{\Lambda}$ in (T1.2) & Fundamental Volume\\
(T1.4)\; $\mathbf{v}_i$ & $n_{i} \times 1$ vector $\mathbf{v}_i \in
\Lambda_{i}$ consists of the elements of $\mathbf{v}$ in
$\Lambda_{i}$, where $\mathbf{v}=[\mathbf{v}_{1}^{T}, \ldots,
\mathbf{v}_{K+1}^{T}]^{T}$ is $\left(\sum_{i=1}^{K+1}n_{i}\right)
\times 1$, and $\Lambda_{i}$ is
transmitter $i$'s lattice (coding or shaping)& Vector for transmitter $i$, where $1 \leq i \leq K$ correspond to the users while $i=K+1$ corresponds to the relay\\
(T1.5)\; $\mathbf{v}_u$, $\mathbf{v}_r$ & $\mathbf{v}_u=[\mathbf{v}_{1}^{T},...,\mathbf{v}_{K}^{T}]^{T}$, $\mathbf{v}_r=\mathbf{v}_{K+1}$, with $\mathbf{v}_i$  defined in (T1.4)& Super-vector for all users, and vector for relay\\
(T1.6)\; $C^{Lo}_{ur}$ &  $C^{Lo}_{1}\times\cdots \times C^{Lo}_{K+1}$, where $C^{Lo}_{i}$ is the Loeliger linear code for transmitter $i$ as in Definition \ref{Def_logier}  & Super Loeliger-linear-code of users and relay \\
(T1.7)\; $\Lambda_{C_u}$, $\Lambda_{S_u}$ & $\Lambda_{C_{1}}\times
\cdots \times \Lambda_{C_{K}} $, $\Lambda_{S_{1}}\times \cdots
\times
\Lambda_{S_{K}} $  & Super-coding and shaping lattices of users  \\
(T1.8)\; $\Lambda_{C_r}$, $\Lambda_{S_r}$ & $\Lambda_{C_{K+1}}$, $\Lambda_{S_{K+1}}$ & Super-coding and shaping lattices of relay   \\
(T1.9)\; $\Lambda_{C_{ur}}$, $\Lambda_{S_{ur}}$ &
$\Lambda_{C_{1}}\times \cdots \times \Lambda_{C_{K+1}}$,
$\Lambda_{S_{1}}\times \cdots \times
\Lambda_{S_{K+1}}$ &  Super-coding and shaping lattices of users and relay\\
(T1.10)\; $\bar{\mathbf{v}}$, {$\bar{\mathbf{v}}_{i}$}& $\mathbf{v} \mod \Lambda_{S_{ur}}$, \; {$\mathbf{v}_{i} \mod \Lambda_{S_{i}}$}  & Modulo lattice operation \\
(T1.11)\; $p_{i},\gamma_{i}$ &  Definition \ref{Def_logier} & Loeliger lattice ensemble parameters \\
(T1.12)\; $\psi^{one}$, $\psi^{mod}$ & Definition \ref{Def_OneOne}, \ref{Def_Modulo} & One-to-one Mapper, Modulo-sum Mapper  \\
(T1.13)\; $C_{u}^{nest}$, $C_{r}^{nest}$ & Definition \ref{Def_OneOne} & Users' Codebooks,  Relay's Codebook \\
(T1.14)\; $\psi^{one}_{\Delta}$, $\psi^{mod}_{\Delta}$ &
$\psi_{\Delta}^{one} :
\psi_{\Delta}^{one}\left(\bar{{\mathbf{d}}}_{u}(\bw)\right)=\bar{{\mathbf{d}}}_{r}(\bw)$,\;
$\psi_{\Delta}^{mod} :
\psi_{\Delta}^{mod}\left(\bar{{\mathbf{d}}}_{u}(\bw)\right)=\bar{{\mathbf{d}}}_{r}(\bw)$
 &  Differential mapper for one-to-one and modulo-sum mapper\\
(T1.15)\;
$(\mathcal{C}_{\mathcal{C}_{\psi_{\Delta},\mathrm{E}}})^{*}$
&$(\mathcal{C}_{\mathcal{C}_{\psi_{\Delta},\mathrm{E}}})^{*}\triangleq\{\bar{\mathbf{d}}:\bar{\mathbf{d}}\in
C_{\psi_{\Delta}}^{*},
C_{\psi_{\Delta}}\in\mathcal{C}_{\psi_{\Delta},\mathrm{E}}\}$ & Differential  codewords in ensemble $\mathcal{C}_{\psi_{\Delta},\mathrm{E}}$ \\
(T1.16)\; $\mathds{O}^{\psi_{\Delta}}$ & \eqref{eq.ambiguity} & Differential ambiguity cosets \\
(T1.17)\; $\mathbf{M}^{S}$ & Matrix
$\mathbf{M}^{S}\triangleq[\mathbf{M}_{i_{1}},...,\mathbf{M}_{i_{|S|}}]$
is formed from $\mathbf{M}=[\mathbf{M}_{1},...,\mathbf{M}_{K_M}]$,
where $K_M$ is the number of the submatrices of $\mathbf{M}$, $S=\{i_{1},...,i_{|S|}\}$, $1 \leq i_{1}<\cdots<i_{|S|} \leq K_M$ \;\;\;\;\;\;\;\;\;\;\;\;\;\;\;\;\;\;\;\;\;\;     & Matrix for users in set $S$\\
(T1.18)\; $ R_{unG}^{dst}(\mathbf{H}_{dst}^{\{S,K+1\}})$ &$\frac{1}{2}\log|\mathbf{I}_{2(|S|M_{u}+M_{r})LT}+\left(\mathbf{H}_{dst}^{\{S,K+1\}}\right)^{H}\mathbf{H}_{dst}^{\{S,K+1\}}|$ & Rate constraint at the destination using unstructured Gaussian codebook\\
(T1.19)\; $ R_{unG}^{{relay}}(\mathbf{H}_{relay}^{S})$ &$\frac{1}{2}\log|\mathbf{I}_{2|S|M_{u}LT}+\left(\mathbf{H}_{relay}^{S}\right)^{H}\mathbf{H}_{relay}^{S}|$& Rate constraint at the relay using unstructured Gaussian codebook\\
(T1.20)\; $d(\mathbf{r})$ & The diversity gain $\lim\limits_{\rho
\rightarrow \infty} \frac{-\log P_{E}(\rho)}{\log \rho}$ given a
certain multiplexing gain $\mathbf{r}$, where $P_{E}(\rho)^{\$}$ is
the probability that not all users are correctly decoded, $\rho$ is
the received SNR, and $\mathbf{r}=[r_{1},...,r_{K}]$ with
$r_{i}\triangleq \lim\limits_{\rho \rightarrow \infty}
\frac{R_{i}(\rho)}{\log \rho}$ and $R_{i}(\rho)$ is the transmission
rate of user
 $i$ & Diversity and multiplexing tradeoff (DMT) \\
(T1.21)\; $\bar{\mathbf{z}}_{p}$& Apply componentwise modulo $p$
operation on
$\mathbf{z}$ & Modulo p\\
(T1.22)\; $\bar{\mathbf{z}}_{\underline{p}}$&
$[\bar{(\mathbf{z}_{1})}_{p_{1}}^{T},...,\bar{(\mathbf{z}_{1})}_{p_{K+1}}^{T}]$,
for  $\underline{p}=(p_{1},...,p_{K+1})$, $\mathbf{z}=[\mathbf{z}_{1}^{T},...,\mathbf{z}_{K+1}^{T}]^{T}$ & Modulo vector $\underline{p}$\\
(T1.23)\; $\underline{\gamma}\mathbf{z}$
 &  $[\gamma_{1}\mathbf{z}_{1}^{T},...,\gamma_{K+1}\mathbf{z}_{K+1}^{T}]^{T}$, for  $\underline{\gamma}=(\gamma_{1},...,\gamma_{K+1})$, $\mathbf{z}=[\mathbf{z}_{1}^{T},...,\mathbf{z}_{K+1}^{T}]^{T}$ & ``Vector" Hadamard  product \\

 \bottomrule[0.8pt]
\end{tabular}
\begin{tablenotes}  
\footnotesize \item[\$]  Instead of $P_{E}(\rho)$, the outage
probability is used for the calculation of DMT of  the relay node in
the DDF \cite{on_the_achievableDMT,optimality_ARQ}
\end{tablenotes}
\end{threeparttable}
\end{table}
\renewcommand{\baselinestretch}{1.785}
then $\psi(\bar{\mathbf{c}}_u)=\bar{\mathbf{c}}_{r}$. Now we
introduce the proposed mappers.  The first one is as follows.
\begin{definition} [One-to-one mapper] \label{Def_OneOne}
The one-to-one  mapper $ \psi^{one}:C^{nest}_u\rightarrow
C^{nest}_r$ for  O-MLC is a one-to-one bijective mapping that maps
coset leaders in the super-codebook of users $C_u^{nest}$ to the
relay codebook $C_r^{nest}$. Here $
C_u^{nest}\triangleq\{\bar{\mathbf{c}}_{u}:\bar{\mathbf{c}}_{u}=(\mathbf{c}_{u}
\mod \Lambda_{S_u}), \mathbf{c}_{u} \in \Lambda_{C_u}\}$ and
$C_r^{nest}\triangleq\{\bar{\mathbf{c}}_{r}:
\bar{\mathbf{c}}_{r}=(\mathbf{c}_{r} \mod \Lambda_{S_{r}}),
\mathbf{c}_{r} \in \Lambda_{C_{r}}\}$, where $\Lambda_{S_u}$ and
$\Lambda_{C_u}$ are defined  in (T1.7) while $\Lambda_{S_r}$ and
$\Lambda_{C_r}$ are defined in (T1.8) in Table
\ref{notations_table}.
\end{definition}

Note that $|C_r^{nest}|=|C_u^{nest}|$ since the aforementioned
mapping is bijective. The one-to-one relay mapper may require high
complexity as the size of super-user codebook $|C_u^{nest}|$ becomes
large. To reduce the complexity of the mapper, we introduce another
mapping $\psi^{mod}$, where the modulo-sum operation is performed at
the relay, which is motivated by the XOR operations in  network
coding \cite{netif}.
\begin{definition} [Modulo-sum mapper] \label{Def_Modulo}
The modulo-sum mapper $\psi^{mod}:C_u^{nest}\rightarrow C_r^{nest}$
for  MS-MLC is defined as
$\psi^{mod}(\bar{\mathbf{c}}_{u})=\sum_{i=1}^{K}\psi^{mod}_{i}(\bar{\mathbf{c}}_{i})\mod
\Lambda_{S_{r}}$, where $\psi^{mod}_{i}: {C_{i}^{nest}}\rightarrow
C_r^{nest}$ is an injective mapping for user $i$ with nested user
codebook $C_{i}^{nest}$ given in Definition \ref{Def_nested}, while
$C_u^{nest}$ and $C_r^{nest}$ are given in Definition
\ref{Def_OneOne}.
\end{definition}

Note we require that $|C_r^{nest}| \geq \max_{i}\{|C_i^{nest}|\}$ to
ensure that the mapping $\psi_{i}^{mod}$ in Definition
\ref{Def_Modulo} is injective. The domain dimension  of $\psi^{mod}$
is at most $\max_{i}\{|C_i^{nest}|\}$ while that of the one-to-one
mapper $\psi^{one}$ is $\prod_{i=1}^{K}|C_i^{nest}|$, and
$\psi^{mod}$ has less complexity compared with $\psi^{one}$.
However, the one-to-one mapper $\psi^{one}$ ensures that two
different users' super-codewords are mapped to different codewords
at the relay, and results in better error performance. In contrast,
it is possible that two different super-codewords map to the same
codeword of the relay due to the modulo-sum operation in
$\psi^{mod}$, and ambiguity occurs while decoding. \vspace{-6mm}
\subsection{Encoders and Proposed $K$-stage Coset Decoders}\label{sec.lattice_coding} {\emph{1) Encoders at the $K$ transmitters and the relay}:
\vspace{-0mm} User $i$ selects the codeword $\bar{\mathbf{c}}_{i}$
according to its message $w_{i}$ from the codebook described in
Section \ref{sec_latticecodebook}, and sends signal $\mathbf{x}_{i}$
into the MARC \eqref{eq_y_relay}-\eqref{eqYdst} (cf.
\eqref{eq_x_dst}) \vspace{-0mm}
\begin{equation}
\mathbf{x}_{i}=\left([\bar{\mathbf{c}}_{i}-\mathbf{u}_{i}] \mod
\Lambda_{S_{i}}\right) \label{eq.8}
\end{equation}
where  $\mathbf{u}_{i}$ is a dither signal uniformly distributed
over the Voronoi region $\mathcal{V}_{\Lambda_{S_{i}}}$ of the
shaping lattice $\Lambda_{S_{i}}$
 ((T1.2) in Table
\ref{notations_table}).
From \cite{erezlattice}, due to the dither $\mathbf{u}_{i}$,
$\mathbf{x}_{i}$ is uniformly distributed over
$\mathcal{V}_{\Lambda_{S_{i}}}$ and independent of
$\bar{\mathbf{c}}_{i}$. To meet the input power constraints
(\ref{eqd.5a}) as in \cite{latc}, we let the second-order moment of
the shaping lattice $\sigma^{2}(\Lambda_{S_{i}})=1/2$.

As for the relay (transmitter $K+1$), it will first decode the
users' messages, using the operation  introduced below. Then the
relay selects its codeword $\bar{\mathbf{c}}_{K+1}$ according to the
{decoded}  transmitted codewords $\bar{\mathbf{c}}_{i}$s using the
mappers in Section \ref{sec_mapper}, and then transmits
$\mathbf{x}_{K+1}$ as in \eqref{eq.8} with  the power constraint
(\ref{eqd.5a}).

{\emph{2) $K$-stage coset decoder:} We first introduce the decoder
at the destination, which generalizes the single stage coset decoder
in \cite{latticemac} to the multi-stage one. The coset decoder
disregards the boundaries of the codewords and avoids the
complicated boundary control\cite{onmax},  which allows for
significant complexity reductions compared to  ML decoding.
Moreover, it facilitates the efficient sphere decoding algorithm
\cite{aunified},\cite{onmax}. To decode messages from the received
signal $\mathbf{y}_{dst}$ in \eqref{eqYdst}, the proposed $K$-stage
coset decoder works as in Table \ref{table.tree} with the detailed
steps explained as follows.
\begin{table}
\begin{threeparttable}
\renewcommand{\arraystretch}{0.9}
\caption{ Algorithm of the $K$-stage coset decoding for the
Destination\tnote{*} }
 \label{table.tree} \centering
\begin{tabular}{ll}
\hline \hline \\
\multicolumn{1}{l}{ \textit{\textbf{A. Generation of the decoding tree:}}} \\
\multicolumn{1}{l}{ \;\textit{\textbf{Initialization:}} For the root node, Node{\_}user$=\text{empty}$,  Node{\_}{stage}$=1$} \\
\;\textbf{for $k=1:K-1$}\\
\;\;\;for each node with  node{\_}stage$=k$,\\
\;\;\;generates $(K-k+1)$ child nodes for the next stage (Node{\_}stage$=k+1$),\\
\;\;\;\;\;for the child nodes  from left to right,\\
\;\;\;\;\;Node{\_}user  are assigned from the set
\;\;\;\;\;$\{1,...,K\}\setminus \mathcal{S}$ in increasing order,
\\\;\;\;\;\;where $\mathcal{S}= \{j: \text{Node{\_}user}=j , \text{for the
ancestors{**}  of child
node} \}$\\
\;\textbf{end}\\
\multicolumn{1}{l}{\textit{\textbf{B. $K$-stage candidate generation via coset decoding:}}} \\
\textbf{for} $k=1:K$\\
\;\;\textbf{\emph{Step B.1}}: For the node $(k,j)$, let
$\mathbf{y}_{dst}^{(k,j)}=\mathbf{y}_{dst}-\sum_{i\in
\mathcal{S}_p^{(k,j)}}{\mathbf{H}_{i}^{d}}\hat{\mathbf{x}}_{i}$,
where $\mathcal{S}_p^{(k,j)}$ is the set of previously-decoded users
$i$  along the
\\\;\;   path from root node to node $(k,j)$,
$ \hat{\mathbf{x}}_{i}$ is the transmitted signal (from
\eqref{eq.8}) corresponding to previous-decoded user $i$'s message,\\
\;\;  and  the channel $\bH_i^{d}$ is formed from \eqref{eqd.7}.
\\\;\; (For example, for the path starting from root node to node
(3,1) in Fig.~\ref{fig.tree}, the set $\mathcal{S}_p^{(3,1)}$ is
$\{1,2\}$.)
\\
\;\;\textbf{\emph{Step B.2}}: Decodes the users' messages in the
residual user set $ \mathcal{S}^{(k,j)}=\{1,...,K\}\setminus
\mathcal{S}_p^{(k,j)}$
by coset decoding \\
\;\; $ \;\;\;\;\;\;\;\;\;\;\;\;\;\;\;\;\;
\;\;\;\;\;\;\;\;\;\;\;\;\;\;\;\;\;\;\;\; \;\;\;
\;\;\;\;\;\;\;\;\;\;\;\;\;\;\;\;\;\;\; \;\;\;
\;\;\;\;\;\;\;\;\;\;\;\;\;\;\;\;\;\;\; \hat{\mathbf{c}}^{(k,j)}=\arg
\min_{\mathbf{c}\in
\mathds{O}^{\psi,{(k,j)}}}\mathrm{M}^{(k,j)}(\mathbf{c}^{(k,j)}), $
\\
\;\; where \\
\;\; $\;\;\;\;\;\;\;\;\;\;\;\;\;\; \;\;\;
\;\;\;\;\;\;\;\;\;\;\;\;\;\;\;\;\;\;\; \;\;\;
\;\;\;\;\;\;\;\;\;\;\;\;\;\;\;\;\;\;\; \;\;\;
\mathrm{M}^{(k,j)}(\mathbf{c}^{(k,j)})=\left|\mathbf{F}_{dst}^{(k,j)}\mathbf{y}_{dst}^{(k,j)}+\mathbf{B}_{dst}^{(k,j)}\left(\mathbf{{u}}^{(k,j)}-\mathbf{{c}}^{(k,j)}\right)\right|^{2}.
\;\;\;\;\;\;\;\;\;\;\;\;\;\;\;\;
\;\;\;\;\;\;\;\;\;\;\;\;\;\;\;\;\;\;\;\;\;\;
\;\;\;\;\;\;\;\;\;\;\;\;\;\;\;\;\; \mathrm{ (T2.1)} $\\\;\; Let
$m_{k}=2((K-k+1)M_{u}+M_{r})LT$ and $N'=2NLT$, the $m_{k} \times 1$
$\mathbf{{c}}^{(k,j)}$ is formed from the super-codeword
$\mathbf{c}$ by collecting all
 \\
\;\;  $\mathbf{c}_{i} \in \Lambda_{C_i}$  of transmitter $i$ ((T1.4)
in Table \ref{notations_table}) where $i \in
\{\mathcal{S}^{(k,j)},K+1\}$; the dither signal
$\mathbf{{u}}^{(k,j)}$ is formed from $\mathbf{u}$ similarly; the

 \\\;\; $m_{k}\times N'$ $\mathbf{F}_{dst}^{(k,j)}$ and $m_{k} \times m_{k}$
$\mathbf{B}_{dst}^{(k,j)}$ are the corresponding MMSE -GDFE filters
for $\mathbf{{c}}^{(k,j)}$;  and the searching cosets formed by
\\\;\; previously-decoded users $i\in \mathcal{S}_p^{(k,j)}$ is
\\\;\;\;\;\;\;\;\;\;\;\;\;\;\;\;\;\;\;\;\;\;\;\;\;\;\;\;\;\;\;\;\;\;\;\;\;
$\mathds{O}^{\psi,{(k,j)}}=\{\mathbf{c}:\mathbf{c}\in
\mathds{O}^{\psi}, ({\mathbf{c}}_{i} \mod
\Lambda_{S_{i}})=(\hat{\mathbf{c}}^{p,(k,j)}_{i} \mod
\Lambda_{S_{i}}), i\in \mathcal{S}_p^{(k,j)}\}$, \\\;\; where
$\mathds{O}^{\psi}$  is given in \eqref{eq_D_c}, $\bc_i \in
\Lambda_{C_i}$ ((T1.4) in Table \ref{notations_table}) where
$\Lambda_{C_i}$ is transmitter $i$'s coding lattice, and
$(\hat{\mathbf{c}}^{p,(k,j)}_{i} \mod \Lambda_{S_{i}})$ is   \\
\;\; the codeword (coset leader) of the previously-decoded user $i$
where $i \in \mathcal{S}_p^{(k,j)}$.

\\ \textbf{end}\\

\multicolumn{1}{l}{\textit{\textbf{C. Candidate elimination:}}} \\

For node $(K,j)$ at the final stage $K$, combine the decoded
messages to produce the $K\times 1$ super-message
$\hat{\mathbf{w}}_{t}^{({K},j)}$ as \\ the candidate at node
$(K,j)$. The decoder searches for all ${K!}$ candidates
$\hat{\mathbf{w}}_{t}^{(K,j)}$ and declares  the one such that
$\mathbf{H}_{dst}\hat{\mathbf{x}}^{(K,j)}$ is \\ nearest to the
received signal $\mathbf{y}_{dst}$ as the final decoded message,
where $\hat{\mathbf{x}}^{(K,j)}$ is the transmitted signal according
to message $\hat{\mathbf{w}}_{t}^{(K,j)}$.\\
\hline \hline
\end{tabular}
\begin{tablenotes}  
\footnotesize \item[*] The algorithm for the relay can be
identically obtained by ignoring the relay's codewords. \item[**]
The ancestors of a node are all the nodes along the path from the
root to that node (not included).
\end{tablenotes}
\end{threeparttable}
\end{table}

According to Table \ref{table.tree}, the decoder first generates the
decoding tree as in Step A. An example for $K=3$ is given in Fig.
\ref{fig.tree}.
\begin{figure}[t] 
\centering
\includegraphics[scale=0.51]{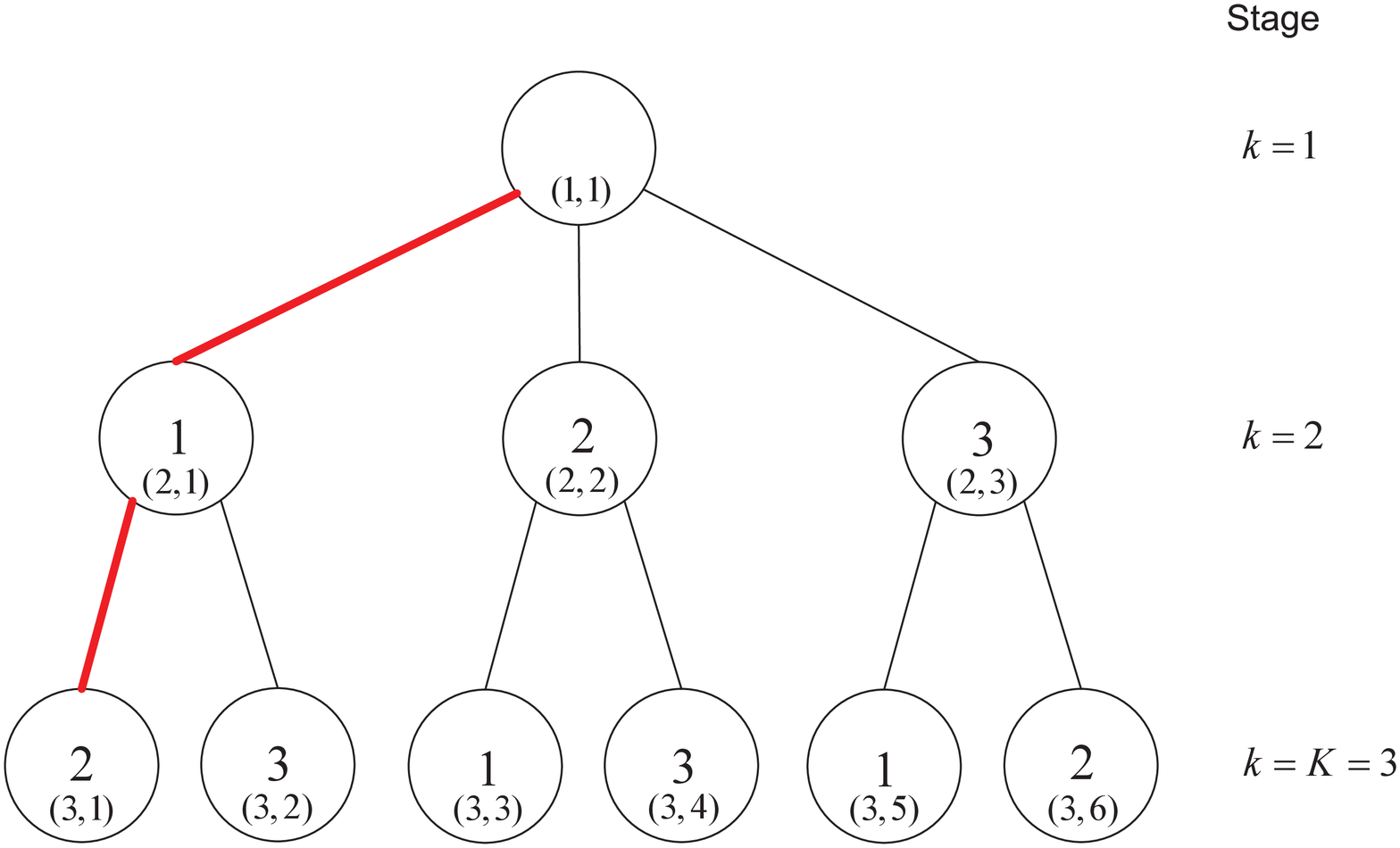}
\caption{The multiuser decoding tree for the $K$-stage coset
decoding in Table \ref{table.tree} with $K=3$. Here for each node,
the label $(k,j)$ denotes the $j$-th node from the left at the
$k$-th stage (Node{\_}stage in Table \ref{table.tree}), while the
number $i$ inside a circle denotes  the  index $i$ of the user
assumed to have been correctly decoded at the previous stage
(Node{\_}user in Table \ref{table.tree}). For example, when the
coset decoding in Table \ref{table.tree} is performed at node
$(2,1)$ (the leftmost child node of the root node), user 1 is
assumed to have been correctly decoded. The path from root node
$(1,1)$ to node $(3,1)$ is illustrated with bolder lines. }
\label{fig.tree}
\end{figure}
The decoder will traverse nodes from stage $1$ to $K$ in the tree,
and produce the candidate codewords. We take the root node in Fig.
\ref{fig.tree} as an example to explain Steps B.1 and B.2 in Table
\ref{table.tree}. We use the notation
$\bar{\mathbf{c}}(\mathbf{w}_t)$ to represent the super-codeword for
the $K+1$ transmitters corresponding to the $K \times 1$ transmitted
message vector $\bw_{t}=[(\bw_{t})_{1},\ldots,(\bw_{t})_{K}]^T$,
where $(\bw_{t})_{i}$ denotes  the transmitted message  for user
$i$. For the root node (the first stage coset decoding), with
received signal $\mathbf{y}_{dst}$ at the destination; the decoder
output $\hat{\mathbf{c}}$ according to (\ref{eq.8}) is
\begin{equation}
\hat{\mathbf{c}}=\arg \min_{\mathbf{c}\in {
\mathds{O}^{\psi}}}\mathrm{M}({\mathbf{c}}), \;\; \mbox{with} \;\;
\mathrm{M}({
\mathbf{c}})=|\mathbf{F}_{dst}\mathbf{y}_{dst}+\mathbf{B}_{dst}({\mathbf{u}}-{\mathbf{c}})|^{2},
\label{eq.4}
\end{equation}
 where
$\mathbf{F}_{dst}$ and $\mathbf{B}_{dst}$  are  the forward and
feedback filters of the minimum mean-square  error (MMSE) estimation
generalized decision feedback equalizer (GDFE) as defined in
\cite{latticemac} and \cite{latc}   respectively;
$\mathbf{u}=[\mathbf{u}_{1}^{T},...,\mathbf{u}_{K+1}^{T}]^{T}$ and
the decoder searches points in the cosets $\mathds{O}^{\psi}$ (see
\eqref{eq_D_c}) of all $\bar{\mathbf{c}}(\mathbf{w})$ (defined
similarly to $\bar{\mathbf{c}}(\mathbf{w}_{t})$ above),
$\mathbf{w}\in \mathcal{W}$, with $ \mathcal{W}$ being the set of
all possible messages:
\begin{equation}
{ \mathds{O}^{\psi}}\triangleq\{\mathbf{c}\in
\Lambda_{C_{ur}}:(\mathbf{c}\mod
\Lambda_{S_{ur}})=\bar{\mathbf{c}}(\mathbf{w}), \mathbf{w}\in
\mathcal{W} \}, \label{eq_D_c}
\end{equation}
where the super-lattice of users and the relay $\Lambda_{C_{ur}}$ is
defined in (T1.9) of Table \ref{notations_table}. The decoded
message $\hat{\mathbf{w}}_t$ is declared if
{$\bar{\mathbf{c}}(\hat{\mathbf{w}}_t)$} and  the decoded
$\hat{\mathbf{c}}$ from \eqref{eq.4}
 belong to the same coset,
$\hat{\mathbf{c}} \mod
\Lambda_{S_{ur}}=\bar{\mathbf{c}}(\hat{\mathbf{w}}_t)$. For the node
$(k,j)$ in the decoding tree (the $j$-th node from the left at
$k$-th stage) we consider a path from the root node to node $(k,j)$.
An example for $(k,j)=(3,1)$ is given in Fig. \ref{fig.tree}. In
Step B.1 of Table \ref{table.tree}, the decoder assumes that all the
users at the nodes along the path (users $1$ and $2$ for the example
path in Fig. \ref{fig.tree}), have already been successively decoded
(not necessarily correctly), and subtract the corresponding
transmitted signals from the received signal $\by_{dst}$.
Then the decoder decodes the remaining transmitted messages by
(T2.1) in the Step B.2 of Table \ref{table.tree} (which corresponds
to \eqref{eq.4}). Finally, as in Step C, the  decoder  searches for
all ${K!}$ candidates produced at the nodes at the $K$-th stage
(instead of all $2^{LT\sum_{i=1}^{K}R_i}$ codewords) to choose the
final decoded message.

The decoder at the relay also uses \eqref{eq.4} as the criterion to
decode messages from $\by_{relay}$ in \eqref{eq_y_relay} with the
corresponding MMSE-GDFE forward and feedback filters. The main
difference is that now the decoding does not make use of the relay
codebook, and the decoder searches in the super-lattice of users
$\Lambda_{C_{u}}$ instead of the coset $\mathds{O}^{\psi}$ in
\eqref{eq_D_c}. The complexity of the decoder in Table
\ref{table.tree} is about
$\bigg(\sum_{k=1}^{K}\frac{K!}{(K-k+1)!}\mathcal{O}(m_{k}^{3})+K!\mathcal{O}\big(m_{K}^{2}\big)\bigg)=\mathcal{O}((LT)^3)$\footnote{
According to our practical design in  Section \ref{Simulation}, one
can use a linear mapper to implement  O-MLC and MS-MLC. Then the
$m_{k}$-dimensional coset decoder can be implemented by the sphere
decoding algorithm in \cite{onmax} with complexity rougly being
$\mathcal{O}(m_{k}^{3})$.}, where $m_k=2((K-k+1)M_u+M_r)LT$. It is
much smaller compared with the complexity of the ML decoder
$\mathcal{O}(2^{LT\sum_{i=1}^{K}R_i})$, which grows exponentially
with the block length $LT$.

Note that since the super-codewords have to satisfy the relay
mapping rule (which may not be linear) in Section \ref{sec_mapper},
the set $\mathds{O}^{\psi}$ is \textit{not} necessary a sublattice
of $\Lambda_{C_{ur}}$. This makes (\ref{eq.4}) different from the
MMSE-GDFE \textit{lattice} decoder in \cite{latticemac} and
\cite{latc}. Without the algebraic structure of a lattice,  the
upcoming error probability analysis in the next section, and the
design of practical decoding algorithms for the simulations in
Section \ref{Simulation} will be much more difficult than those in
\cite{latticemac}.

\vspace{-3mm}
\section{ Performance analysis of the proposed coding schemes }
\label{sec.Error_Analysis} \vspace{-2mm}
  In this section, we establish the achievable rate regions for the
  MARC defined in (\ref{eq_y_relay})  and (\ref{eqYdst}),  using the proposed O-MLC and
MS-MLC for a fixed channel matrix, respectively. We show that the
rate performance, which was originally achieved by using an
unstructured random codebook in \cite{cooperaive_capacity}, is now
achieved by our structured O-MLC. The key is using the $K$-stage
coset decoder which performs successive cancellation on the
multiuser decoding tree,  thus avoiding the rate loss incurred by
the one-stage coset decoder in \cite{latticemac}.  The rate loss due
to use of a one-stage coset decoder is derived in Corollary
\ref{coro_ach_one}. However, in Corollary \ref{DMTthm1}, we show
that the rate loss is relatively small in the high SNR regime, and
structured O-MLC with the one-stage coset decoder   achieves the
optimal DMT for the MARC in
 (\ref{eq_y_relay}) and  (\ref{eqYdst}). Note that the DMT was
 achieved by an unstructured random codebook and ML decoding in
 \cite{optimality_ARQ}. For  MS-MLC, we show that it
can approach the rate performance of  O-MLC  by increasing the
relay's codebook size, and thus can tradeoff between the rate
performance and complexity.

In the error  analysis of the proposed schemes, the conventional
approach tailored  for  ML decoding \cite{case_MAC}
\cite{optimality_ARQ} and \cite{DMT_MAC} fails in predicting the
performance of the coset decoder in (\ref{eq.4}) due to the infinite
number of points $\mathbf{c} \in \mathds{{O}}^{\psi}$  where the set
$\mathds{{O}}^{\psi}$ is defined in (\ref{eq_D_c}). To solve this
problem, from (\ref{eq_D_c}), we define the \emph{  differential
ambiguity  cosets} for the event that the transmitted message $
\mathbf{w}_{t}$ is erroneously decoded as $\mathbf{w}$ as
\begin{equation}
{\mathds{O}^{\psi_\Delta}}\triangleq\{\mathbf{d}\in
\Lambda_{C_{ur}}:\bar{\mathbf{d}}={\bar{\mathbf{d}}(\mathbf{w})} ,
\mathbf{w}\in\mathcal{W}, \mathbf{w}\neq\mathbf{w}_t\},
\label{eq.ambiguity}
\end{equation}
where the \emph{differential codeword}
$\bar{\mathbf{d}}(\mathbf{w})\triangleq
(\bar{\mathbf{c}}(\mathbf{w})-\bar{\mathbf{c}}(\mathbf{w}_t)\mod
\Lambda_{S_{ur}})$ with $\Lambda_{S_{ur}}$ given in (T1.9) of Table
\ref{notations_table} and the vector after modulo operation
$\bar{\mathbf{d}}$ is defined in (T1.10). From the closure property
of lattice addition, $\bar{\mathbf{d}}(\mathbf{w})\in
\Lambda_{C_{ur}}$. Moreover, $\mathds{O}^{\psi_\Delta}$ is not  a
direct product of $K+1$ lattices (i.e., $\Lambda_{C_{ur}}$), and
thus the techniques in \cite{latticemac} fail to predict the error
probability of O-MLC. {We propose a new error probability
upper-bound which avoids directly counting points of
$\mathds{O}^{\psi_\Delta}$ in the decision region of the decoder as
this kind of evaluation is intractable.} Please see the upcoming
Lemma \ref{MHthm1} presented in the proof of Theorem
\ref{Thm_ach_one} and the discussions after it.

Besides providing the aforementioned new proof techniques, we also
show that there will be a rate loss due to the one-stage coset
decoding in \cite{latticemac}. The loss can be circumvented with the
proposed $K$-stage coset decoders by letting the decoder
successively cancel the previously decoded messages. We show that in
our multiuser decoding tree as in Fig. \ref{fig.tree}, there exists
at least one path at each stage of Step B of Table \ref{table.tree}
on which the previously decoded messages are correct. Then we can
\emph{at least} obtain a better decoder for the remaining users in
the next stage to improve the error performance.  To show that we
can always choose the correct codeword from the candidates at the
final stage in the decoding tree, we use a suboptimal decoder
instead of the optimal one in Step C of Table \ref{table.tree} to
complete our  proof. Note that our decoder is different from the
successive MAC decoding studied in \cite{macratesplit}, where the
decoder is based on ML decoding and the previously decoded messages
are correct.

Now,  we are ready to derive the achievable rate region of
(\ref{eq_y_relay}) and (\ref{eqYdst}),   using  O-MLC as follows.
\begin{thm}
\label{Thm_ach_one} For  the MARC in (\ref{eq_y_relay}) and
(\ref{eqYdst}), the DDF rate region in (\ref{ach_one_relay}) and
(\ref{ach_one}), which is achieved by unstructured Gaussian
codebooks and ML decoding in \cite{cooperaive_capacity}, is
achievable by the \textit{structured} O-MLC and the $K$-stage coset
decoder in Table \ref{table.tree}, where the rate constraints at the
relay and destination are
\begin{align}
\sum_{i\in S}R_{i}< &
\frac{1}{ LT}R_{unG}^{relay}(\mathbf{H}_{relay}^{S}), \;\text{and}\;\;\;\label{ach_one_relay}\\
 \sum_{i\in S}R_{i}< &
\frac{1}{ LT}R_{unG}^{dst}(\mathbf{H}_{dst}^{\{S,K+1\}})
,\;\;\;\;\forall S \subseteq \{1,...,K\} \label{ach_one}
\end{align}
respectively, with $R_{unG}^{dst}(\mathbf{H}_{dst}^{\{S,K+1\}})$ and
$R_{unG}^{relay}(\mathbf{H}_{relay}^{S})$  given in (T1.18) and
(T1.19) in Table \ref{notations_table}. The channel matrix from the
users in the set $S$ and the relay to the destination
$\mathbf{H}^{\{S,K+1\}}_{dst}$ is formed from
$\mathbf{H}_{dst}=[\mathbf{H}^d_1,\ldots ,\mathbf{H}^d_{K+1}]$ as in
(T1.17) with $\mathbf{H}^d_i$ given in \eqref{eqd.7} and
\eqref{eq.relaychannel}, and the channel matrix from the users in
the set $S$ to the relay $\mathbf{H}_{relay}^{S}$ is defined
similarly to $\mathbf{H}^{\{S,K+1\}}_{dst}$.

\end{thm}
\begin{proof}
We  prove only (\ref{ach_one}) here since (\ref{ach_one_relay})
follows similarly. First, for the first stage ($k=1$, the root node
of Fig. \ref{fig.tree}) of the candidate generation process in Step
B of Table \ref{table.tree}, we show that at least one of the users'
messages is  correctly decoded in the generated ``super''-message
$\hat{\mathbf{w}}_{t}^{(1,1)}$ of all users (with probability 1) as
$T\rightarrow\infty$. To do this, we first define the following
error event.
\begin{definition} [set-${S}$ error] \label{Def_setS}
A decoded super-message $\mathbf{w}$ is with set-${S}$ error if the
message in $\mathbf{w}$ for every user $i$, where $i\in {S}$, is
different from the corresponding transmitted message. That is,
$\bw_i \neq (\bw_t)_{i}, \forall i \in {S}$, while $\bw_i =
(\bw_t)_{i}$, otherwise.
\end{definition}

\noindent Let $P_{e}({S}|\mathbf{H}_{dst})$ be the probability that
there exists $\mathbf{w}$ with set-${S}$ error with fixed
$\mathbf{H}_{dst}$, and $\min_{\bc \in
\mathbf{o}(\mathbf{w})}\mathrm{M}(\bc)\leq\min_{\bc \in
\mathbf{o}(\mathbf{w}_{t})}\mathrm{M}(\bc)$, with
$\mathrm{M}(\mathbf{c})$ defined in \eqref{eq.4} and $
\mathbf{o}(\mathbf{w})$ being the coset of $\mathbf{w}$. To validate
our claim, we first consider the erroneous user set
$\mathcal{S}^{(1)}=\{1,...,K\}$ and will prove that
$P_{e}(\mathcal{S}^{(1)}|\mathbf{H}_{dst}) \rightarrow 0$ for the
first-stage, if the transmission rates $R_{i}$ satisfy
(\ref{ach_one}) and the lattice codes are good as defined in the
upcoming Lemma \ref{MHthm1}. Here
$P_e(\mathcal{S}^{(1)}|\mathbf{H}_{dst})$ is averaged over { the
random relay-mapper and linear-code ensemble
$\mathcal{E}_{\psi,C^{Lo}}=\{\psi^{one},C^{Lo}_{ur}\}$ consisting of
all possible one-to-one mappers $\psi^{one}$ and Loeliger linear
codes $C^{Lo}_{ur}$ of the users and relay ((T1.6) in Table
\ref{notations_table}).}
\begin{lemma}
\label{MHthm1} For O-MLC, let $R_i,$ $i=1,...,K+1$, be the code
rates for the users and the relay, and $\{\Lambda_{C_{i}}\}$ belong
to the Loeliger lattices ensembles (cf. Definition \ref{Def_logier}
in Appendix \ref{MHthm1proof}-(I)). For stage $k=1$ of Step B in
Table \ref{table.tree}, as $LT\rightarrow\infty$, the
set-$\mathcal{S}^{(1)}$ error probability (cf. Definition
\ref{Def_setS}), where $\mathcal{S}^{(1)}=\{1,...,K\}$, satisfies
\begin{equation}
\begin{split}
& P_e(\mathcal{S}^{(1)}|\mathbf{H}_{dst}) \leq
\frac{1}{|\mathcal{E}_{\psi,C^{Lo}}|}\sum_{(\psi^{one},C^{Lo}_{ur})\in
\mathcal{E}_{\psi,C^{Lo}}}\left|\mathds{O}_{\mathcal{S}^{(1)}}^{\psi^{one}_{\Delta}}\cap
\mathcal{R}_{\beta}\right| \\ & \leq \exp\Bigg(\frac{-LT}{\log
e}\Bigg[\frac{1}{LT} R_{unG}^{dst}(
\mathbf{H}_{dst}^{\{\mathcal{S}^{(1)},K+1\}})-\sum_{i\in
\mathcal{S}^{(1)}}R_{i}
+\frac{1}{LT}\log\frac{2^{R_{K+1}LT}-1}{2^{R_{K+1}LT}}\Bigg]
\Bigg) \label{Minkowmarc2}
\end{split}
\end{equation}
where $\mathds{O}_{\mathcal{S}^{(1)}}^{\psi^{one}_{\Delta}}$
consists of points belonging to the differential ambiguity cosets
for O-MLC $\mathds{O}^{\psi^{one}_{\Delta}}$ (cf.
\eqref{eq.ambiguity}), with corresponding messages having
set-$\mathcal{S}^{(1)}$ errors; $
\mathcal{E}_{\psi,C^{Lo}}=\{\psi^{one},C^{Lo}_{ur}\}$ is defined
right before Lemma \ref{MHthm1}; the decision region
$\mathcal{R}_{\beta}\triangleq
\left\{\mathbf{v}:|\mathbf{B}_{dst}\mathbf{v}|^{2}\leq
(KM_{u}+M_{r})LT(1+\beta)\right\} $ with filter $\mathbf{B}_{dst}$
defined as in \eqref{eq.4} and $\beta>0$; and the rate constraint
$R_{unG}^{dst}(\mathbf{H}_{dst}^{\{\mathcal{S}^{(1)},K+1\}})$ is
defined as in \eqref{ach_one}.
\end{lemma}
The proof of Lemma \ref{MHthm1} is given in Appendix
\ref{MHthm1proof}. The main difficulty is that the cosets
$\mathds{O}^{\psi^{one}_{\Delta{}}}$ is not a direct product of
$K+1$ lattices as in \cite{latticemac},  so the methods in
\cite{latticemac} and \cite{avgboundlattice} \textit{cannot} be
directly applied to counting the number of points of
$\mathds{O}_{\mathcal{S}^{(1)}}^{\psi^{one}_{\Delta}}$ in the
decision region $\mathcal{R}_{\beta}$ in the second inequality of
\eqref{Minkowmarc2}. We avoid explicitly counting points in
$\mathds{O}_{\mathcal{S}^{(1)}}^{\psi^{one}_{\Delta}}$ by developing
new upper-bounds as in \eqref{thm1pf0} and \eqref{thm1pf2} in
Appendix \ref{MHthm1proof}. Otherwise, naively applying the methods
of  \cite{latticemac} and \cite{avgboundlattice} will result in
rates as in \eqref{Minkowmarc2} but without the factor
${(2^{R_{K+1}LT}-1)}$ cancelling out $2^{R_{K+1}LT}$, and lead to
significant rate loss compared with our \eqref{ach_one} with
${S}=\mathcal{S}^{(1)}$ since $R_{K+1}=\sum_{i=1}^{K}R_{i}$ is
required  to ensure the one-to-one mapping.

With the results for the first stage $k=1$ in Lemma \ref{MHthm1}, we
show by induction that after the candidate generation process in
Step B of Table \ref{table.tree}, among all ``super''-message
$\hat{\mathbf{w}}_{t}^{(K,j)}$ at  stage $K$ (defined in Step C),
there exists a correct one almost surely (with probability 1) as
$T\rightarrow\infty$.  To do this, we will show that for stage $k$,
with at least $k-1$ (almost surely) correctly decoded users from the
previous stage, almost surely there exists one node $(k,j_{k}')$
having at least $k$ users correctly decoded. Note that for stage
$k$, conditioned on the event that all decoded users' messages from
the previous stages are correct, the noise $\mathbf{n}_{dst}$ in
\eqref{eqYdst} may \textit{no longer} be Gaussian
\cite{macratesplit}. However, under the condition  \eqref{ach_one},
the probability $P^{(k)}_e$ that there exists no node  at stage $k$
having at least $k$ users correctly decoded can  be shown to still
satisfy
\begin{equation} \label{eq_Pek}
P^{(k)}_e \overset{(a)}{\leq}
P_e(\mathcal{S}^{(1)}|\mathbf{H}_{dst})+\sum_{s=2}^{k}P^G_{e}(\mathcal{S}^{(s,j'_{s})}|\mathbf{H}_{dst},\mathcal{S}_{p}^{(s,j'_{s})})
\overset{(b)}{\rightarrow} 0,
\end{equation}
as $LT \rightarrow \infty$, where
$P^G_{e}(\mathcal{S}^{(s,j'_{s})}|\mathbf{H}_{dst},\mathcal{S}_{p}^{(s,j'_{s})})$
is defined under Gaussian $\mathbf{n}_{dst}$ and will be given below
and (\ref{eq_Pek} a) follows from \cite{macratesplit}. Then our
claim for stage $K$ is valid  and $P^{(K)}_e \rightarrow 0$ by
induction. Since under \eqref{ach_one}, as $LT \rightarrow \infty$,
$P_e(\mathcal{S}^{(1)}|\mathbf{H}_{dst}) \rightarrow 0$ from
\eqref{Minkowmarc2}, we will show that
$P^G_e(\mathcal{S}^{(s,j'_{s})}|\mathbf{H}_{dst},\mathcal{S}_{p}^{(s,j'_{s})})\rightarrow
0, \forall s \leq k$, {  in this setting} to validate (\ref{eq_Pek}
b). Now we introduce the definition of
$P^G_e(\mathcal{S}^{(s,j'_{s})}|\mathbf{H}_{dst},\mathcal{S}_{p}^{(s,j'_{s})})$
as follows. Let $\mathcal{S}_{p}^{(s,j'_{s})}$ be the set of $s-1$
previous users along the path starting from the root node to node
$(s,j'_s)$, the $j'_s$-th node at stage $s$, in the decoding tree
shown in Fig. \ref{fig.tree}. Also, let the set
$\mathcal{S}^{(s,j'_{s})}$ be
$\{1,\ldots,K\}\setminus\mathcal{S}_{p}^{(s,j'_{s})}$. Then
$P^G_e(\mathcal{S}^{(s,j'_{s})}|\mathbf{H}_{dst},\mathcal{S}_{p}^{(s,j'_{s})})$
is defined as the probability  that there exists $\mathbf{w}$ with
set-$\mathcal{S}^{(s,j'_{s})}$ error (Definition \ref{Def_setS})
conditioned on the event that all users in
$\mathcal{S}_{p}^{(s,j'_{s})}$ are correct (the existence of
$j'_{s}$ is guaranteed by the assumption of induction
$P_{e}^{(s-1)}\rightarrow0$, $1<s \leq k$), and $\mathbf{n}_{dst}$
in \eqref{eqYdst} is conditionally Gaussian. For this kind of error
events,  $\min_{\bc \in
\mathbf{o}(\mathbf{w})}\mathrm{M}^{(s,j'_s)}(\bc)\leq\min_{\bc \in
\mathbf{o}(\mathbf{w}_{t})}\mathrm{M}^{(s,j'_s)}(\mathbf{c})$ with
$\mathrm{M}^{(s,j'_s)}(\mathbf{c})$ defined on the right-hand side
(RHS) of (T2.1) in Table \ref{table.tree}. As in the proof for Lemma
\ref{MHthm1} in Appendix \ref{MHthm1proof}, we can  similarly
upper-bound
$P^G_e(\mathcal{S}^{(s,j'_s)}|\mathbf{H}_{dst},\mathcal{S}_{p}^{(s,j'_s)})$
by the RHS of (\ref{Minkowmarc2}) with $\mathcal{S}^{(1)}$ replaced
by $\mathcal{S}^{(s,j'_s)}$. Thus if the transmission rates $R_{i}$
satisfy (\ref{ach_one}), as $T\rightarrow\infty$, we have that
$P^G_e(\mathcal{S}^{(s,j'_s)}|\mathbf{H}_{dst},\mathcal{S}_{p}^{(s,j'_s)})\rightarrow
0$, which verifies (\ref{eq_Pek} b). This validates  our claim for
stage $K$.

For the Step C of Table \ref{table.tree}, we will use the following
suboptimal decoder instead of the optimal decoder in Table
\ref{table.tree} to prove that we can find the correct message
$\bw_t$ almost surely. First, we compare candidates
$\hat{\mathbf{w}}_{t}^{(K,1)}$ and $\hat{\mathbf{w}}_{t}^{(K,2)}$,
and form the set of users $\mathcal{S}_{c}$ so that for any $i \in
\mathcal{S}_c$, $\hat{\mathbf{w}}_{t}^{(K,1)}$ and
$\hat{\mathbf{w}}_{t}^{(K,2)}$ have a common message for user $i$.
Then we compare the ``coset''-distances
$\min_{\mathbf{c}\in\mathbf{o}(\hat{\mathbf{w}}_{t}^{(K,1)})}\mathrm{D}^{(k_c)}({\mathbf{c}})$
and
$\min_{\mathbf{c}\in\mathbf{o}(\hat{\mathbf{w}}_{t}^{(K,2)})}\mathrm{D}^{(k_c)}({\mathbf{c}})$
of these two candidates and choose the one with smaller
``coset''-distance (if equal, we randomly select one), where
$\mathrm{D}^{(k_c)}(\bc)$ is formed by replacing
$\mathcal{S}_p^{(k,j)}$ with $\mathcal{S}_{c}$ in
$\mathrm{M}^{(k,j)}(\bc)$ in (T2.1) of Table \ref{table.tree} (also
the corresponding parameters). We then compare the chosen candidate
with the next candidate $\hat{\mathbf{w}}_{t}^{(K,3)}$, and so on.
After $K!-1$ comparisons among total $K!$ candidates, the final
chosen candidate in the final comparison will be declared as the
decoded message. Now we show that the error probability of the above
sub-optimal decoder will approach zero. As in (\ref{eq_Pek} a), this
error probability is upper-bounded by
$P^{(K)}_e+P^G_{e}(\hat{\bw}_t^{(K,j)}|\hat{\bw}_t^{(K,
j'_{K})}=\bw_t)$, where $P^{(K)}_e$ is defined before \eqref{eq_Pek}
and $P^G_{e}(\hat{\bw}_t^{(K,j)}|\hat{\bw}_t^{(K,j'_K)}=\bw_t)$ is
the probability that the sub-optimal decoder outputs incorrect
$\hat{\bw}_t^{(K,j)} \neq \bw_t$ conditioned on the event that there
is one correct candidate $\hat{\bw}_t^{(K,j'_K)}=\bw_t$ and the
noise $\mathbf{n}_{dst}$ is Gaussian. Since $P^{(K)}_e \rightarrow
0$ according to the previous paragraph, we will show
$P_{e}(\hat{\bw}_t^{(K,j)}|\hat{\bw}_t^{(K,j'_K)}=\bw_t) \rightarrow
0$ and then our proof is complete. Specifically, if the decoder
output $\hat{\bw}^{(K,j)} \neq \bw_t$, it will have smaller (or
equal) ``coset''-distance than that of $\hat{\bw}^{(K,j'_K)}=\bw_t$,
i.e.,
$\min_{\mathbf{c}\in\mathbf{o}(\hat{\mathbf{w}}_{t}^{(K,j)})}\mathrm{D}^{(k_c)}({\mathbf{c}})
\leq
\min_{\mathbf{c}\in\mathbf{o}(\hat{\bw}^{(K,j'_K)})}\mathrm{D}^{(k_c)}({\mathbf{c}})$,
and now $\mathcal{S}_{c}$ becomes the set of correctly decoded users
in $\hat{\bw}^{(K,j)}$ since it is the set of users with common
messages for both $\hat{\bw}_t^{(K,j)}$ and the correct
$\hat{\bw}^{(K,j'_K)}=\mathbf{w}_{t}$. That is, message
$\hat{\bw}_t^{(K,j)}$ may have set-$(\mathcal{S}_{c})^{c}$ error
(Definition \ref{Def_setS}) given that the users in
$\mathcal{S}_{c}$ are correct, where
$(\mathcal{S}_{c})^{c}=\{1,\ldots,K\}\setminus \mathcal{S}_{c}$.
However, from the derivations in the previous paragraph, conditioned
on the event that  users in $\mathcal{S}_{c}$ are correct, the
probability of set-$(\mathcal{S}_{c})^{c}$ error $
P^G_e((\mathcal{S}_{c})^{c}|\mathbf{H}_{dst},\mathcal{S}_{c})
\rightarrow 0$. This is a contradiction and
$P^G_{e}(\hat{\bw}_t^{(K,j)}|\hat{\bw}_t^{(K,j'_K)}=\bw_t)
\rightarrow 0$. Thus our suboptimal decoder will always find the
correct $\mathbf{w}_t$, and this concludes our proof since the
optimal decoder in Table \ref{table.tree} will perform even better.
\end{proof}

If only the one-stage coset decoder is used as in \cite{latticemac},
we have the following.
\begin{corollary}
\label{coro_ach_one} For the MARC in (\ref{eq_y_relay}) and
(\ref{eqYdst}), the rate region constrained by
\eqref{ach_one_relay_onestage} and \eqref{ach_one_1stage}, which is
strictly smaller than that in Theorem \ref{Thm_ach_one}, is
achievable by  O-MLC with the one-stage coset decoder in
\eqref{eq.4}, where
\begin{align}
\sum_{i\in S}R_{i}< &
\frac{1}{LT}R_{unG}^{relay}(\mathbf{H}_{relay}^{S})-M_{u}|S|\log\frac{K}{|S|}
\;\; \mbox{and} ,\;\;
\label{ach_one_relay_onestage} \\
 \sum_{i\in S}R_{i}< &
\frac{1}{LT}R_{unG}^{dst}(\mathbf{H}_{dst}^{\{S,K+1\}})-(M_{u}|S|+M_{r})\log\frac{KM_{u}+M_{r}}{|S|M_{u}+M_{r}}
,\;\;\;\;\forall S \subseteq \{1,...,K\}. \label{ach_one_1stage}
\end{align}
\end{corollary}
The proof can be easily  obtained by modifying Lemma \ref{MHthm1},
in which we count all of the  points in cosets
$\mathds{O}^{\psi^{one}_{\Delta{}}}$ instead of only counting those
corresponding to the message with set-$\mathcal{S}^{(1)}$
 error (Definition \ref{Def_setS}) as in
$\mathds{O}_{\mathcal{S}^{(1)}}^{\psi^{one}_{\Delta}}$ of
\eqref{Minkowmarc2}, and follows arguments similar to those used in
Theorem \ref{MHthm1}. The details are omitted here. Clearly,
compared to the rate region in (\ref{ach_one_relay}) and
(\ref{ach_one}), there are rate loss terms
$M_{u}|S|\log\frac{K}{|S|}$ and
$(M_{u}|S|+M_{r})\log\frac{KM_{u}+M_{r}}{M_{u}|S|+M_{r}}$ in
(\ref{ach_one_relay_onestage}) and (\ref{ach_one_1stage}),
respectively. These losses are zero when $|S|=K$, and the MMSE-GDFE
processing for the one-stage coset decoding in \eqref{eq.4} is only
sum rate optimal.

For  MS-MLC, we have the following theorem. In this result, in
addition to the same rate constraints (\ref{ach_one_relay}) and
(\ref{ach_one}) as in Theorem \ref{Thm_ach_one}, there is an
additional rate constraint (\ref{ach_modsum2}) for MS-MLC which
makes the achievable rate region smaller than that for O-MLC.
\begin{thm}
\label{Thm_mod_sum} For the MARC in (\ref{eq_y_relay}) and
(\ref{eqYdst}), using MS-MLC and the $K$-stage coset decoder in
Table \ref{table.tree}, the rate region with constraints in
(\ref{ach_one_relay}) and (\ref{ach_one}) and the following
additional constraint (\ref{ach_modsum2}) is achievable, where
\begin{equation}
\sum_{i\in S}R_{i}<
\frac{1}{LT}R_{unG}^{dst}(\mathbf{H}_{dst}^{S})-M_{u}|S|\log\frac{|S|M_{u}+M_{r}}{|S|M_{u}}+
R_{K+1} \;\;\;\forall S \subseteq \{1,...,K\}, |S|>1
\label{ach_modsum2}.
\end{equation}
When using MS-MLC with one-stage coset decoder in \eqref{eq.4}, the
rate region with the constraints in \eqref{ach_one_relay_onestage}
and (\ref{ach_one_1stage}) and the following additional constraint
\eqref{ach_modsum2_1stage} is achievable, where
\begin{equation}
\sum_{i\in S}R_{i}<
\frac{1}{LT}R_{unG}^{dst}(\mathbf{H}_{dst}^{S})-M_u|S|\log\frac{KM_{u}+M_{r}}{|S|M_{u}}+R_{K+1}
\;\;\;\forall S \subseteq \{1,...,K\}, |S|> 1.
\label{ach_modsum2_1stage}
\end{equation}

\end{thm}
\begin{proof}
Unlike O-MLC, there is a possibility for MS-MLC that two different
users' super-codewords are mapped to the same  relay codeword from
Definition \ref{Def_Modulo}. This fact makes the properties
exploited in Lemma \ref{MHthm1} for the random mapped-codebook
ensemble of O-MLC (for details, please see the proof of
\eqref{eq.avgmapper2} in Appendix \ref{MHthm1proof})
 no longer hold for the ensemble for
MS-MLC. Thus Lemma \ref{MHthm1} cannot be applied for MS-MLC. We
solve this problem by dividing the random mapped-codebook
 ensemble for MS-MLC into two partitions, and the techniques
for proving Lemma \ref{MHthm1} can be modified to deal with each
partition separately. The detailed proof is given in  Appendix
\ref{appendix_sketchpfMTH2}. The rate region for one-stage coset
decoder in \eqref{ach_one_relay_onestage}, (\ref{ach_one_1stage})
and \eqref{ach_modsum2_1stage} follows by using techniques similar
to those used in the proof of Corollary \ref{coro_ach_one}.
\end{proof}
The additional rate constraint (\ref{ach_modsum2}) is due to the
ambiguity of the modulo-sum mapper in MS-MLC, where there is a rate
loss term $M_{u}|S|\log\frac{|S|M_{u}+M_{r}}{|S|M_{u}}$. However,
the rate constraint (\ref{ach_modsum2}) can be negligible and even
looser than constraint (\ref{ach_one}), as
$\big(R_{K+1}-M_{u}|S|\log\frac{|S|M_{u}+M_{r}}{|S|M_{u}}\big)$
becomes larger by increasing the relay codebook size $2^{R_{K+1}LT}$
(which reduces the occurrence of ambiguity). Thus  MS-MLC can
approach the performance of O-MLC by increasing the complexity.

Finally, for random slow fading channels, we show that  O-MLC with
the one-stage coset decoder \eqref{eq.4} is DMT optimal for the DDF
MARC, as stated in the following corollary. Despite the rate loss
terms in (\ref{ach_one_relay_onestage}) and (\ref{ach_one_1stage})
compared with (\ref{ach_one_relay}) and (\ref{ach_one}),
respectively, the losses become relatively negligible for the DMT
analysis when the SNR is high.

\begin{corollary}
\label{DMTthm1} For the MARC in (\ref{eq_y_relay}) and
(\ref{eqYdst}), with the one-stage coset decoder \eqref{eq.4}, the
O-MLC achieves the optimal DDF DMT $d(\mathbf{r})$ of
(\ref{eq_y_relay}) and (\ref{eqYdst}), respectively, where
$d({\mathbf{r}})$ is defined in (T1.20) of Table
\ref{notations_table}.
\end{corollary}

\textit{Sketch of proof:} As in \cite{on_the_achievableDMT} and
\cite{optimality_ARQ}, we need to establish the DMT optimality for
both the relay and destination channels. We focus on the destination
channel (\ref{eqYdst}) since the DMT-optimality for the relay
channel (\ref{eq_y_relay}) (identical to the MAC channel)  has been
proved in \cite{latticemac}. Following \cite{latc} and the proof
steps for \eqref{ach_one_1stage}, we can exponentially upper-bound
the error probability $P_e(\rho_d)$ in (T1.20) of Table
\ref{notations_table} using decoder \eqref{eq.4} (averaged over
random $\mathbf{H}_{dst}$ which satisfy (\ref{ach_one})) as
\begin{equation}
P_e(\rho_d)\dot{\leq}E_{\mathbf{H}_{dst}}\left[(1+\delta) \cdot
\sum_{S\subseteq\{1,...,K\},S\neq\phi}\rho_d^{LT\sum_{i\in
S}r_i}\exp\left[\frac{-1}{\log
e}R_{unG}^{dst}\left(\mathbf{H}_{dst}^{\{S,K+1\}}\right)\right]\right]\doteq
Pr(\mathcal{O}) \label{eq.DMTanalsis1}
\end{equation}
where $\delta> 0$, $\rho_d$ is the received SNR at the destination;
$r_i$ is the given multiplexing gain for user $i$ as in (T1.20); the
exponential larger and equal \cite{DMT_MAC} are denoted as
$\dot{\geq}$ and $\dot{=}$; and $\mathcal{O}$ is the outage event
when $\mathbf{H}_{dst}$ does not satisfy (\ref{ach_one}). The proof
of
(\ref{eq.DMTanalsis1}) is detailed in Appendix {\ref{appendixthmdmtone}}. 
Together with the fact that for any coding schemes, $P_e(\rho_d)
\dot\geq Pr(\mathcal{O})\doteq\rho_{d}^{-d(\mathbf{r})}$ as in
\cite{DMT_MAC},  we prove that O-MLC can achieve the optimal DMT
$d(\mathbf{r})$ for the destination. \hfill $\blacksquare$

In \cite{optimality_ARQ}, a two-user, single antenna node MARC was
studied for the symmetric rate case ($R_{1}=R_{2}$), which showed
that the DDF strategy achieves the optimal DMT for the MARC in the
low to medium multiplexing gain regime. The DMT results of Corollary
\ref{DMTthm1} can be achieved by codebooks, which are more
structured than that in \cite{optimality_ARQ}. Moreover, our designs
in the next section also demonstrate that our theoretical results
can be implemented in practice.

\vspace{-8mm}
\section{Simulation Results}\label{Simulation}
\vspace{-2mm} In this section, we present numerical examples  to
illustrate our theoretical  results. Performance results based on
practical decoders are also presented. As mentioned in Section
\ref{sec.lattice_coding}, the \emph{ lattice decoder} in
\cite{latticemac} and \cite{latc} fails to be directly applicable to
our \emph{coset decoder } of (\ref{eq.4}) since only the  points in
$\mathds{O}^{\psi}$ of (\ref{eq_D_c}) will be searched. In general,
the optimal non-linear relay mapper may make the coset decoders very
complicated and impractical. To facilitate the coset decoder for the
relay mapper, we resort to the sub-optimal \emph{linear} mapper such
that the coset decoder of (\ref{eq.4}) can be transformed into the
efficient lattice decoder. For simplicity, we consider the case in
which there are two users with the same transmission rate, i.e,
$R_{1}=R_{2}=R$.

{ Let the code rate of the relay $R_{3}=2R$, and $\mathbf{G}_{{i}}$,
$i=1,2,3$, be the generation matrix of the coding lattice
$\Lambda_{C_i}$ (cf. Definition \ref{Def_logier}) for transmitter
$i$. Then for user $i=1,2$, the codewords are
$\bar{\mathbf{c}}_{i}=(\mathbf{G}_{{i}}\tilde{\mathbf{z}}_{i} \mod
\Lambda_{S_{i}})$, where $\tilde{\mathbf{z}}_{i}\in
\mathbb{Z}^{2M_uLT}$. For O-MLC, with $M_{r}=2M_{u}$, we choose the
linear relay mapping such that the relay codewords are
$\bar{\mathbf{c}}_{3}=(\mathbf{G}_{{3}}\tilde{\mathbf{z}}_{3} \mod
\Lambda_{S_{3}})$ with
$\tilde{\mathbf{z}}_{3}=[\tilde{\mathbf{z}}_{1}^{T},\tilde{\mathbf{z}}_{2}^{T}]^{T}$.}
After some manipulations, it can be verified that the decoding
equation of (\ref{eq.4})  is transformed into
\begin{equation}
\hat{\mathbf{z}}=\arg \min_{\mathbf{z}\in
\mathbb{Z}^{n}}|\mathbf{F}_{dst}\mathbf{y}_{dst}+(\mathbf{B}_{dst}\mathbf{u}-\mathbf{B}_{dst}\mathbf{Gz})|^{2}
\label{eq.latticedecoder}
\end{equation}
where $n=8M_{u}LT$. Then for the linear  one-to-one relay mapper, we
have
\begin{equation}
\mathbf{G}=diag(\mathbf{G}_{{1}},\mathbf{G}_{{2}},\mathbf{G}_{{3}})\cdot
\begin{pmatrix}
\begin{array}{cc}
\mathbf{I}_{2M_{u}LT} & \mathbf{0} \\
 \mathbf{0} &  \mathbf{I}_{2M_{u}LT} \end{array} & \mathbf{0} \\
\mathbf{I}_{4M_{u}LT} & 2^{\frac{R}{2M_{u}}}\mathbf{I}_{4M_{u}LT}
\end{pmatrix}.
\label{moddecoder}
\end{equation}
 For the linear modulo-sum relay mapper, with
$M_u=M_{r}$, we choose the linear relay mapping such that
$\tilde{\mathbf{z}}_{3}=\tilde{\mathbf{z}}_{1}+\tilde{\mathbf{z}}_{2}$
and the corresponding $\mathbf{G}$ can be similarly derived.
Note now that the decoder searches the whole integer vector plane
$\mathbb{Z}^{n}$ in (\ref{eq.latticedecoder}), thus the lattice
decoder using the  efficient sphere decoding algorithm
\cite{aunified},\cite{onmax} can be applied .

In the following simulation results,  the number of slots is
selected as $L=2$, and  the sum rate $(R_{1}+R_{2})$, is $4$ BPCU.
The relay forwards the message only when the users' messages are
correctly decoded. All the channel links  are Rayleigh faded and
unless otherwise specified,  the sources-to-relay (S-R) channel link
is 10 dB better than the other channel links. In Fig.
\ref{fig.singlecompare},
\begin{figure}[!t] 
\centering
\includegraphics[scale=0.74]{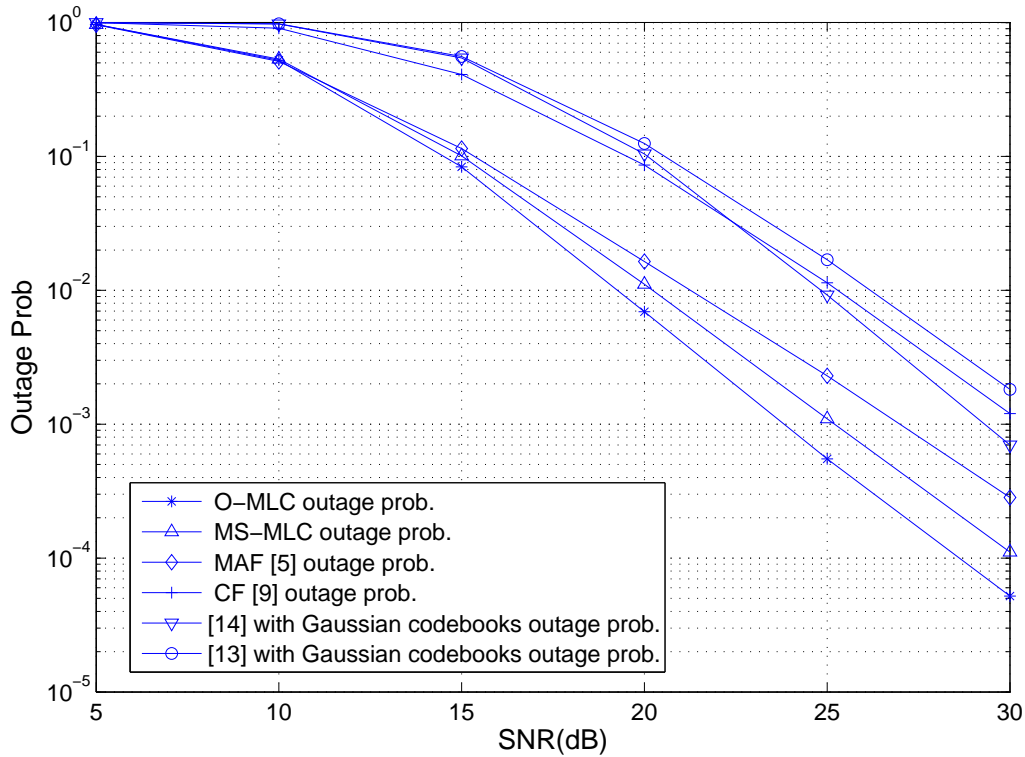}
\caption{The outage probability for  O-MLC (\ref{ach_one_relay}),
(\ref{ach_one}) and  MS-MLC (\ref{ach_one_relay}), (\ref{ach_one}),
(\ref{ach_modsum2}) vs. the protocols in \cite{case_MAC},
\cite{cooperative_wirelss}, \cite{MARCnetowrk2} and
\cite{complexnetwork}.  } 
\label{fig.singlecompare}
\end{figure}
for single-antenna nodes, we show that O-MLC has better error
performance than that of  MS-MLC and both outperform the protocols
of \cite{case_MAC}, \cite{cooperative_wirelss}, \cite{MARCnetowrk2}
and \cite{complexnetwork} in terms of outage probability and achieve
the diversity $\min\{M_{u}(M_{r}+N),(M_{u}+M_{r})N\}$ as expected.
In Fig. \ref{fig.multicompare},
\begin{figure}[t] 
\centering
\includegraphics[scale=0.74]{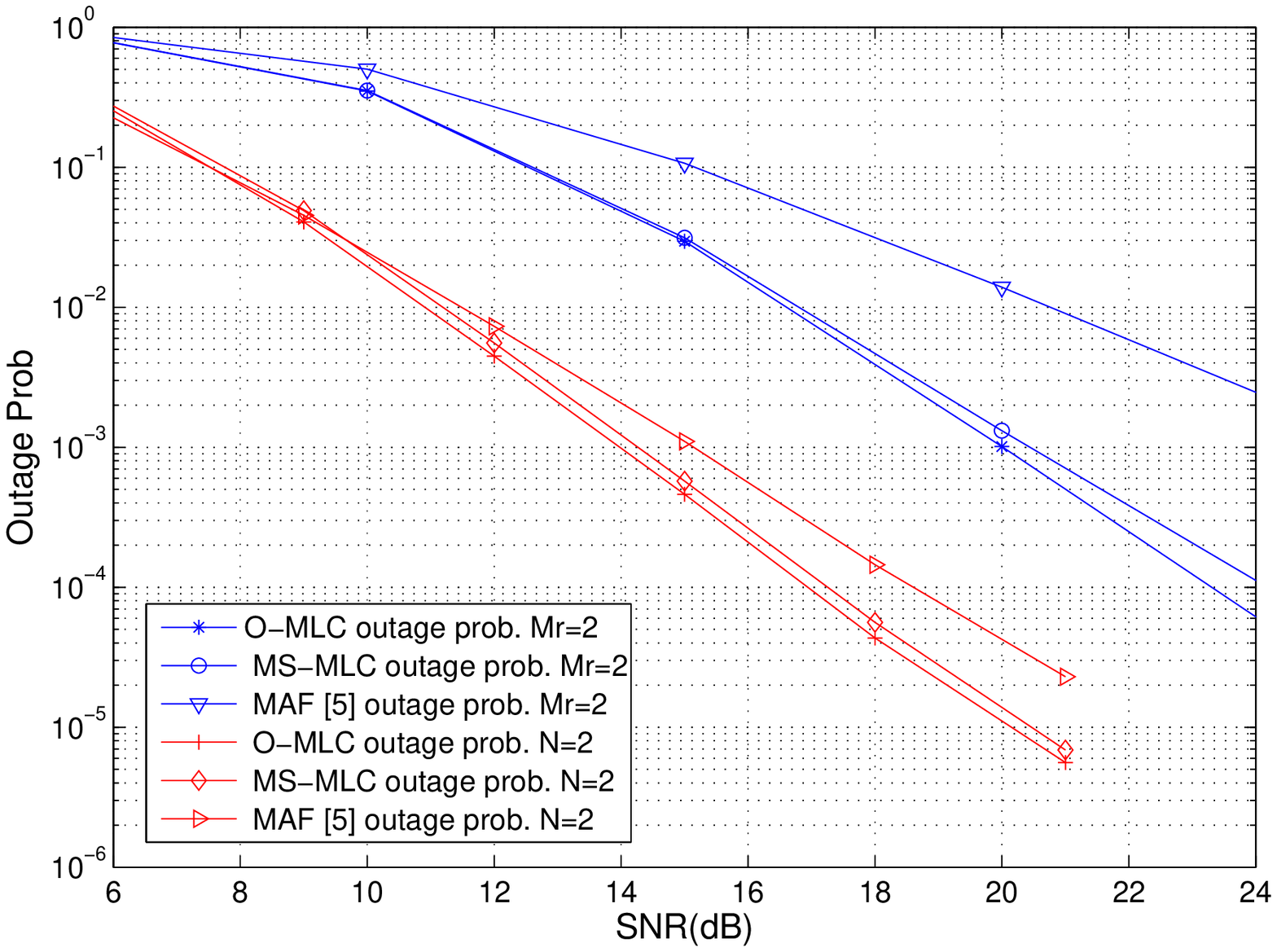}
\caption{The outage probability for  O-MLC (\ref{ach_one_relay}),
(\ref{ach_one}) and  MS-MLC (\ref{ach_one_relay}),
(\ref{ach_one}), (\ref{ach_modsum2}) vs.  MAF \cite{case_MAC}.  } 
\label{fig.multicompare}
\end{figure}
\begin{figure}[!h] 
\centering
\includegraphics[scale=0.7]{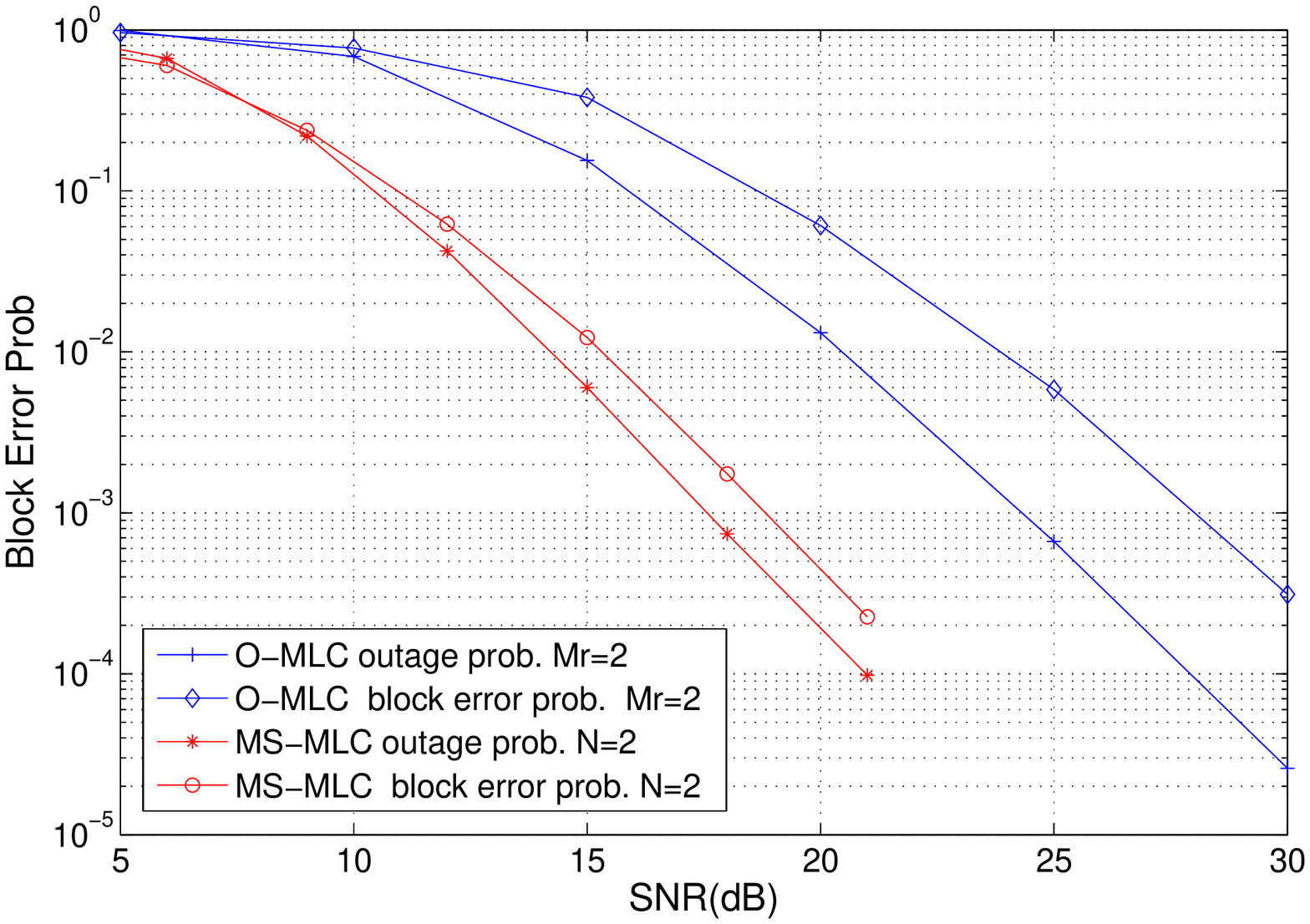}
\caption{Comparison of theoretical outage probabilities  and the
block error probabilities using practical linear relay mapping,
(\ref{ach_one_relay_onestage}), (\ref{ach_one_1stage}) for  O-MLC
and (\ref{ach_one_relay_onestage}),
(\ref{ach_one_1stage}), (\ref{ach_modsum2_1stage}) for  MS-MLC .   } 
\label{fig.latticecompare}
\end{figure}
for the cases $M_{u}=N=1,M_{r}=2$ and $M_{u}=M_{r}=1,N=2$ (where the
S-R link is 15 dB better than the other channel links),
respectively, we show that our proposed coding schemes outperform
the MAF. For the former case, the MAF achieves a diversity of only 2
instead of 3. Note the methods in \cite{cooperative_wirelss},
\cite{MARCnetowrk2} and \cite{complexnetwork}
 cannot be straightforwardly extended to
the   case of multiple-antenna nodes.

 For the simulation of practical lattice codings based on
one-stage practical decoder and linear relay mapper, with  the slot
length $T=2$, we use the pair of self-similar randomly generated
nested lattices drawn from the lattice ensemble defined in
Definition \ref{Def_logier}.
For the settings the same as the above, in Fig.
\ref{fig.latticecompare}, the block error rate for  O-MLC and MS-MLC
are presented. The parameters of the linear codes in the lattice
ensemble for  O-MLC and  MS-MLC are $(p_{i},k_{i}) =(97,3),(47,3),
\forall i$ (cf. Definition \ref{Def_logier}), respectively.
The diversity of $3$ for each user is achieved as expected using
our finite $T$ code construction.

\vspace{-4mm}
\section{Conclusion}
\vspace{-2mm} \label{sec_conclusion} In this work, we have proposed
O-MLC and MS-MLC for structured MARC coding.  The former enjoys
better error performance, while the latter provides more flexibility
to tradeoff between the complexity and the error performance. The
error performance of  MS-MLC can approach that of  O-MLC by
increasing the complexity. We have shown that with the new $K$-stage
decoding instead of the one-stage decoding  considered in previous
works, the structured O-MLC can approach the rate performance of
unstructured codebook with ML decoding. When only the one-stage
decoder is used,  O-MLC can still achieve the optimal DMT of DDF.
Besides the theoretical results, we have also considered the design
of practical short length lattice code with \emph{linear} mapping,
which facilitates the efficient lattice decoding.
 Simulation results have shown
that our proposed coding schemes outperform existing schemes in
terms of outage probabilities. \vspace{-4mm}
\appendix

\vspace{-4mm} \vspace{-3mm}
\subsection{Proof of Lemma \ref{MHthm1}    } \label{MHthm1proof}
\vspace{-3mm} \textbf{(I) Some useful definitions :} Here we
introduce some notation for simplification. We denote the nesting
ratio in Definition \ref{Def_nested} as $\tau_i=2^{R_i/2M_u}$ while
the dimensions of the lattice code are $n_i=2M_uLT$, $(1 \leq i \leq
K)$. The corresponding parameters for the relay are $\tau_{K+1}$ and
$n_{K+1}$, respectively. We also have the following definitions.

\begin{definition} [Loeliger lattices ensemble \cite{avgboundlattice}] \label{Def_logier}
Let $\bar{\Lambda}_{C_i}$ be a lattice generated by a linear code $
C^{Lo}_i$ as
$\bar{\Lambda}_{C_i}\triangleq\{\mathbf{z}\in\mathbb{Z}^{n_i}:\bar{\mathbf{z}}_{p_i}\in
C^{Lo}_{i}\}$, where $\bar{\mathbf{z}}_{p_i}$ is obtained by
applying the componentwise reduction modulo $p_i$ operation on
$\mathbf{z}$ \cite{avgboundlattice} and the $(n_i,k_i)$ linear code
$C^{Lo}_i$ is defined over the finite field $\mathbb{Z}_{p_i}^{n_i}$
((T1.1) in Table \ref{notations_table}). The Loeliger lattices
ensemble is the lattices ensemble
$\{\Lambda_{C_i}=(\gamma_i\bar{\Lambda}_{C^{Lo}_i}):C^{Lo}_i\in
\mathcal{C}_{i,Loe}, \gamma_i\in\mathbb{R}\}$, where
$\mathcal{C}_{i,Loe}$ is a balanced set of linear codes $C^{Lo}_i$
\cite{avgboundlattice}. In our analysis, we let
$p_i\rightarrow\infty$, and
 $\gamma_i\rightarrow 0$  such that
the fundamental volume of $\Lambda_{C_i}$ ((T1.3) in Table
\ref{notations_table})
$V_{f}(\Lambda_{C_i})=p_{i}^{n_i-k_i}\gamma_{i}^{n_i}$ is fixed.
\end{definition}

The following balanced set definition generalizes the balanced set
defined in\cite{avgboundlattice}.

\begin{definition}\emph{(Balanced set for the $K$-user MARC):} \label{Def_balance}
 Let $C$ be the set of
 $\mathbf{c}$ where $\mathbf{c}=\mathbf{c}_{1}\times\cdots \times\mathbf{c}_{K+1}
\in \mathbb{R}^{n_{1}}\times\cdots \times \mathbb{R}^{n_{K+1}}$, and
$\mathcal{C}_{\mathrm{E}}$ be the finite set of ${C}$ (e.g.,
$\mathbf{c}$ is a codeword of a codebook $C$, and
$\mathcal{C}_{\mathrm{E}}$ is a codebook ensemble). We collect all
non-zero $\mathbf{c}$ in ${C}$ of $\mathcal{C}_{\mathrm{E}}$ as
$(\mathcal{C}_{\mathcal{C}_{\mathrm{E}}})^{*}\triangleq\{\mathbf{c}\in
\mathbb{R}^{n}:\mathbf{c}\in C^{*}, C \in
\mathcal{C}_{\mathrm{E}}\}$, where $n=\sum_{i=1}^{K+1}n_{i}$,
$C^{*}=C\setminus\{\mathbf{0}\}$. The set $\mathcal{C}_{\mathrm{E}}$
is called \emph{balanced} if every nonzero element $\mathbf{c}$ in
$(\mathcal{C}_{\mathcal{C}_{\mathrm{E}}})^{*}$ is contained in the
same number,  denoted by $N_{b}$, of $C$  from
$\mathcal{C}_{\mathrm{E}}$. We refer to $N_{b}$ as the
\emph{balanced number}.
\end{definition}

\textbf{(II) Proof:} Here we  show the proof only for the second
inequality of \eqref{Minkowmarc2} since the proof for the first one
is similar to that in
\cite{latticemac}. 
An outline of the proof is provided first to provide  insight into
how to solve the problem that cosets with set-$\mathcal{S}^{(1)}$
errors
$\mathds{O}_{\mathcal{S}^{(1)}}^{\psi^{one}_{\Delta}}\triangleq
\{\mathbf{d}\in
\mathds{O}^{\psi^{one}_{\Delta{}}}:\bar{\mathbf{d}}_{i}\neq
\mathbf{0},\forall i \in \mathcal{S}^{(1)} \}$ (or even cosets
$\mathds{O}^{\psi^{one}_{\Delta{}}}$ in \eqref{eq.ambiguity}) is not
a direct product of $K+1$ lattices, where the differential coset
leader $\bar{\mathbf{d}}_{i}$ for user $i$ is defined below
\eqref{eq.ambiguity} with  (T1.4) and (T1.10) in Table
\ref{notations_table}. First, by averaging over the ensemble of
mappers, and judiciously using the  balanced set  property in
Definition \ref{Def_balance}, we can upper bound
$\frac{1}{|\mathcal{E}_{\psi,C^{Lo}}|}\sum_{(\psi^{one},C_{ur}^{Lo})\in
\mathcal{E}_{\psi,C^{Lo}}}\left|\mathds{O}_{\mathcal{S}^{(1)}}^{\psi^{one}_{\Delta}}\cap
\mathcal{R}_{\beta}\right|$ in (\ref{Minkowmarc2}) using the RHS of
(\ref{thm1pf2} b) below. Note that instead of summation over cosets
$\mathds{O}_{ \mathcal{S}^{(1)}}^{\psi^{one}_{\Delta}}$ as in the
RHS of (\ref{thm1pf0}) below, in the RHS of (\ref{thm1pf2} b) the
summation is over the lattice points of set
$(\Lambda_{C_{ur}})^{\star}$ in \eqref{eq_CurStar} below, which
makes further upper-bounding possible. By taking the limits, we
conclude our proof.

Now we give the details to show the second inequality of
\eqref{Minkowmarc2}. First we introduce some useful notation for the
upcoming (\ref{Minkowmarc1} a). The differential mapper
$\psi_{\Delta}^{one}$, which corresponds to $\psi^{one}$ in
Definition \ref{Def_OneOne}, is defined by replacing the
super-codeword $ \bar{\bc}=[(\bar{\bc}_u)^{T},(\bar{\bc}_r)^T]^T$ in
$\psi^{one}$ with the differential super-codeword
$\bar{\mathbf{d}}(\bw)$ in \eqref{eq.ambiguity}, as in (T1.14) of
Table \ref{notations_table}. Let
$\mathcal{E}_{\psi_{\Delta},C^{Lo}}$ be the ensemble corresponding
to $\mathcal{E}_{\psi,C^{Lo}}$ in \eqref{Minkowmarc2}, but with
one-to-one mappers $\psi_{\Delta}^{one}$ replaced by the
corresponding differential mappers $\psi_{\Delta}^{one}$. Also let
$f(\cdot)$ be the indicator function where $f(\mathbf{d})=1$ if
$\mathbf{d} \in \mathcal{R}_{\beta}$, otherwise ${f}(\mathbf{d})=0$.
Clearly the following (\ref{Minkowmarc1} a) is valid for the
left-hand side (LHS) of the second inequality of \eqref{Minkowmarc2}
since
$|\mathcal{E}_{\psi,C^{Lo}}|=|\mathcal{E}_{\psi_{\Delta},C^{Lo}}|$,
\begin{equation} \label{Minkowmarc1}
\small
\frac{1}{|\!\mathcal{E}_{\psi,C^{Lo}}\!|}\!\!\sum_{(\psi^{one},C^{Lo}_{ur})\in
\mathcal{E}_{\psi,C^{Lo}}}\!\!\!\!\!\!\!\!\!\!\!\!\!\!\left|\mathds{O}_{\mathcal{S}^{(1)}}^{\psi^{one}_{\Delta}}\cap
\mathcal{R}_{\beta}\right|\!\!\overset{(a)}{=}\!\!\!\!\!\!\!\!\!\!\sum_{(\psi_{\Delta}^{one},C^{Lo}_{ur})\in
\mathcal{E}_{\psi_{\Delta},C^{Lo}}}\sum_{\mathbf{d}\in\mathds{O}_{\Delta{},\mathcal{S}^{(1)}}^{\psi^{one}}}\frac{f(\mathbf{d})}{|
\mathcal{E}_{\psi_{\Delta},C^{Lo}}|} \!\!\overset{(b)}{\leq}\!\!
\frac{(\tau_{K+1})^{n_{K+1}}\!\!\prod_{i\in\mathcal{S}^{(1)}
}2^{R_{i}LT}\!\!\!\!}{(\tau_{K+1})^{n_{K+1}}-1}\frac{\int_{\mathbb{R}^{\sum_{i=1}^{K+1}n_{i}}}f(\mathbf{d})d\mathbf{d}}{\prod_{i\in
\{\mathcal{S}^{(1)},K+1\}}V_{f}(\Lambda_{S_{i}})}.
\end{equation}
As for the above (\ref{Minkowmarc1} b), it can be proved from the
RHS of the upcoming (\ref{thm1pf2} b) by averaging over
$\mathcal{C}_{Loe}$ using techniques similar to those in
\cite{latticemac} and \cite{avgboundlattice}. Thus we focus on the
proof of (\ref{thm1pf2} b) below. As pointed out in the beginning of
this appendix, our trick to prove this \textit{critical} step is
replacing the summation over the ``non-lattice'' cosets
$\mathds{O}_{\mathcal{S}^{(1)}}^{\psi^{one}_{\Delta}}$ in the LHS of
(\ref{Minkowmarc1} b) with the set $(\Lambda_{C_{ur}})^{\star}$ in
(\ref{thm1pf2} b) by showing
\begin{align}
&\frac{1}{|\mathcal{E}_{\psi_{\Delta},C^{Lo}}|}\sum_{(\psi_{\Delta}^{one},C^{Lo}_{ur})\in
\mathcal{E}_{\psi_{\Delta},C^{Lo}}}\sum_{\mathbf{d}\in\mathds{O}_{\mathcal{S}^{(1)}}^{\psi^{one}_{\Delta}}}f(\mathbf{d})=\frac{1}{|\mathcal{C}_{Loe}|}\sum_{C^{Lo}_{ur}\in\mathcal{C}_{Loe}}\Bigg(\frac{1}{|\mathcal{C}_{\psi_{\Delta},\mathrm{E}}|}\sum_{\psi_{\Delta}^{one}\in\Psi_{\Delta}^{one}}\sum_{\mathbf{d}\in
\mathds{O}_{\mathcal{S}^{(1)}}^{\psi^{one}_{\Delta}}}f(\mathbf{d})
\Bigg)
\label{thm1pf0}\\
\overset{(a)}{=}&\frac{1}{|\mathcal{C}_{Loe}|}\sum_{C^{Lo}_{ur}\in\mathcal{C}_{Loe}}\Bigg(\frac{1}{((\tau_{K+1})^{n_{K+1}}-1)}\sum_{\mathbf{d}\in(\Lambda_{C_{ur}})^{\diamond}}f(\mathbf{d})\Bigg)\overset{(b)}{\leq}
\frac{1}{((\tau_{K+1})^{n_{K+1}}-1)|\mathcal{C}_{Loe}|}\sum_{C^{Lo}_{ur}\in\mathcal{C}_{Loe}}\sum_{\mathbf{d}\in(\Lambda_{{C}_{ur}})^{\star}}f(\mathbf{d}),
\label{thm1pf2}
\end{align}
where the derivation of each step comes as follows:\\
For \underline{\eqref{thm1pf0}}, we first define
$\mathcal{C}_{\psi_{\Delta},\mathrm{E}}$ as the ensemble of all
mapped nested-codebooks (differential) $C_{\psi_{\Delta}^{one}}$
given a particular super Loeliger linear code $C^{Lo}_{ur}$ (T1.6),
with codewords of $C_{\psi_{\Delta}^{one}}$ satisfying the mapping
rules of the corresponding $\psi_{\Delta}^{one}$. Note that all
$C_{\psi_{\Delta}^{one}} \in \mathcal{C}_{\psi_{\Delta},\mathrm{E}}$
are based on the same $C^{Lo}_{ur}$, but with different mappers.
Also let
$\mathcal{C}_{Loe}=\mathcal{C}_{1,Loe}\times\cdots\times\mathcal{C}_{K+1,Loe}$
be the ensemble of all possible $C^{Lo}_{ur}$ with
$\mathcal{C}_{i,Loe}$ given in Definition \ref{Def_logier}, and
$\Psi_{\Delta}^{one}$ be the ensemble of all possible differential
mappers. Then \eqref{thm1pf0} is obtained by
$|\mathcal{E}_{\psi_{\Delta},C^{Lo}}|=|\mathcal{C}_{\psi_{\Delta},\mathrm{E}}|{|\mathcal{C}_{Loe}|}$
by definition.
\\
For \underline{(\ref{thm1pf2} a)}, given mapper
$\psi^{one}_{\Delta}$ and Loeglier linear code $C_{ur}^{Lo}$ (thus
mapped-codebook $C_{\psi_{\Delta}^{one}}$), we rewrite the
set-$\mathcal{S}^{(1)}$ error cosets as
$\mathds{O}_{\mathcal{S}^{(1)}}^{\psi^{one}_{\Delta}}=\{\mathbf{d}\in
\Lambda_{C_{ur}}:  \bar{\mathbf{d}}\in C_{\psi_{\Delta}^{one}}^{*},
\bar{\mathbf{d}}_{i}\neq \mathbf{0}, \forall i \in \mathcal{S}^{(1)}
\}$, where $\Lambda_{C_{ur}}$ is in (T1.9), and
set-$\mathcal{S}^{(1)}$ errors is defined in Definition
\ref{Def_setS}. Then  the term inside the parentheses on the  LHS of
(\ref{thm1pf2} a) comes from
\begin{align} \label{eq.avgmapper2}
\frac{1}{|\mathcal{C}_{\psi_{\Delta},\mathrm{E}}|}\sum_{\psi_{\Delta}^{one}\in\Psi_{\Delta}^{one}}\sum_{\mathbf{d}\in
\mathds{O}_{\mathcal{S}^{(1)}}^{\psi^{one}_{\Delta}}}f(\mathbf{d})=\frac{1}{|\mathcal{C}_{\psi_{\Delta},\mathrm{E}}|}
\bigg(N_b
\sum_{\mathbf{d}\in(\Lambda_{C_{ur}})^{\diamond}}f(\mathbf{d})\bigg)
\end{align}
where we collect all points belonging to cosets
$\mathds{O}_{\mathcal{S}^{(1)}}^{\psi^{one}_{\Delta}}$ over all
possible mapped codebooks $C_{\psi_{\Delta}^{one}}$ as
$(\Lambda_{C_{ur}})^{\diamond}\triangleq \left\{\mathbf{d}\in
\Lambda_{C_{ur}}: \mathbf{d} \in
\mathds{O}_{\mathcal{S}^{(1)}}^{\psi^{one}_{\Delta}}, {
C_{\psi_{\Delta}^{one}}\in
\mathcal{C}_{\psi_{\Delta},\mathrm{E}}}\right\} $. For
\eqref{eq.avgmapper2}, it comes from the fact that
$\mathcal{C}_{\psi_{\Delta},\mathrm{E}}$ is a balanced set as
follows, where
$(\mathcal{C}_{\mathcal{C}_{\psi_{\Delta},\mathrm{E}}})^{*}$ is the
collection of non-zero codewords in
$\mathcal{C}_{\psi_{\Delta},\mathrm{E}}$, by setting
$(\mathcal{C}_{\mathcal{C}_{\mathrm{E}}})^{*}$ in Definition
\ref{Def_balance} with
$\mathcal{C}_{\mathrm{E}}=\mathcal{C}_{\psi_{\Delta},\mathrm{E}}$
((T1.15) in Table \ref{notations_table} ). Consider two different
vectors $\bc$ and $\bc'$ belonging to
$(\mathcal{C}_{\mathcal{C}_{\psi_{\Delta},\mathrm{E}}})^{*}$. For
each mapped-codebook $C_{\psi_{\Delta}^{one}} \in
\mathcal{C}_{\psi_{\Delta},\mathrm{E}}$ containing $\bc$ but not
$\bc'$, with the corresponding mapper $\psi_{\Delta}^{one}$, we can
easily form another $C_{(\psi_{\Delta}^{one})'} \in
\mathcal{C}_{\psi_{\Delta},\mathrm{E}}$ containing $\bc'$ by forming
a new one-to-one mapper $(\psi_{\Delta}^{one})'$ { from
$\psi_{\Delta}^{one}$}. Therefore, $\bc$ and $\bc'$ are symmetric,
and thus each vector in
$(\mathcal{C}_{\mathcal{C}_{\psi_{\Delta},\mathrm{E}}})^{*}$ is
contained in equal number, denoted by $N_b$, of
$C_{\psi_{\Delta}^{one}}$ from
$\mathcal{C}_{\psi_{\Delta},\mathrm{E}}$. Then
$\mathcal{C}_{\psi_{\Delta},\mathrm{E}}$ is a balanced set as in
Definition \ref{Def_balance}. Together with the fact that
$(\mathcal{C}_{\mathcal{C}_{\psi_{\Delta},\mathrm{E}}})^{*}$ is the
set of coset leaders of $(\Lambda_{C_{ur}})^{\diamond}$, that is,
$(\mathcal{C}_{\mathcal{C}_{\psi_{\Delta},\mathrm{E}}})^{*}=\{\bar{\mathbf{d}}:\mathbf{d}\in(\Lambda_{C_{ur}})^{\diamond}\}$,
then \eqref{eq.avgmapper2} follows.  { Finally,  with
$(\tau_{K+1})^{n_{K+1}}$ being the relay codebook size, since the
differential mapper $\psi_{\Delta}^{one}$ is one-to-one, each
nonzero user codeword can possibly be mapped to
$(\tau_{K+1})^{n_{K+1}}-1$  relay codewords. Also the mapped
nested-codebook ensemble $\mathcal{C}_{\psi_{\Delta},\mathrm{E}}$ is
a balanced set with balanced number $N_b$, we have that
$|\mathcal{C}_{\psi_{\Delta},\mathrm{E}}|/N_b=(\tau_{K+1})^{n_{K+1}}-1$.
Then we obtain (\ref{thm1pf2} a) from (\ref{eq.avgmapper2}).}
\\
For \underline{(\ref{thm1pf2} b)}, we define
$(\Lambda_{C_{ur}})^{\star}$  formed from the super coding-lattice
$\Lambda_{C_{ur}}$ ((T1.9) in Table \ref{notations_table}) as
\begin{equation} \label{eq_CurStar}
(\Lambda_{C_{ur}})^{\star}\triangleq \left\{\mathbf{d}\in
\Lambda_{C_{ur}} : \mathbf{d}_{i}\neq \mathbf{0}, \forall
i\in\mathcal{S}^{(1)}\right\}.
\end{equation}
From the definition of $(\Lambda_{C_{ur}})^{\diamond}$ right after
\eqref{eq.avgmapper2}, we have $(\Lambda_{C_{ur}})^{\diamond}
\subset (\Lambda_{C_{ur}})^{\star}$. Together with the fact that the
indicator function $f(\cdot)$, defined right before
\eqref{Minkowmarc1}, is a nonnegative  function, (\ref{thm1pf2} b)
is obtained.

Finally, the second inequality of (\ref{Minkowmarc2}) can be
obtained from (\ref{Minkowmarc1} b) by following steps similar to
those  in \cite{latticemac} and \cite{latc}. The key observation is
that as $T\rightarrow\infty$, the shaping  lattices
$\Lambda_{S_{i}}$ from Definitions \ref{Def_nested} and
\ref{Def_logier} will be good for minimum square error quantization
\cite{latticegoodfor}, so that their Voronoi regions
$V_{f}(\Lambda_{S_{i}})$   will make the signal behave like an
optimal Gaussian signal. Thus the term $\frac{1}{LT} \log
\int_{\mathbb{R}^{\sum_{i=1}^{K+1}n_{i}}}f(\mathbf{d})d\mathbf{d}/\prod_{i\in
\{\mathcal{S}^{(1)},K+1\}}V_{f}(\Lambda_{S_{i}})$ in
(\ref{Minkowmarc1} b) will approach $-\frac{1}{LT}
R_{unG}^{dst}(\mathbf{H}_{dst}^{^{\{\mathcal{S}^{(1)},K+1\}}})$ in
(\ref{Minkowmarc2}). With $(\tau_{K+1})^{n_{K+1}}=2^{R_{K+1}LT}$ as
defined in Appendix \ref{MHthm1proof}-(I), we then have the second
inequality of (\ref{Minkowmarc2}). The details are given
in Appendix {\ref{appendixderiveof16}}. 

\vspace{-6mm}
\subsection{   Proof of the rate region of the $K$-stage MS-MLC in Theorem \ref{Thm_mod_sum}}
\label{appendix_sketchpfMTH2} \vspace{-2mm}

The proof for the rate region of MS-MLC is similar to the proof of
Theorem \ref{Thm_ach_one}. Here we  show only the principal
difference, which results from the fact that the balanced set
structure exploited in Appendix \ref{MHthm1proof} (to obtain
\eqref{eq.avgmapper2}) is no longer valid for  MS-MLC. We solve this
problem by introducing a new 2-partition balanced set  defined in
 Definition \ref{Def_2_balance} below. Specifically, we will show a
counterpart of \eqref{Minkowmarc1} for the first stage as follows:
For MS-MLC, with the relay-mapper and linear-code ensemble
$\mathcal{E}_{\psi,C^{Lo}}$ of $\{\psi^{mod},C_{ur}^{Lo}\}$ and
$\mathcal{S}^{(1)}=\{1,...,K\}$, we have
\begin{equation}
\begin{split}
&\frac{1}{|\mathcal{E}_{\psi,C^{Lo}}|}\sum_{(\psi^{mod},C_{ur}^{Lo})\in
\mathcal{E}_{\psi,C^{Lo}}}\left|\mathds{O}_{\mathcal{S}^{(1)}}^{\psi^{mod}_{\Delta}}\cap
\mathcal{R}_{\beta}\right| \\
&\leq \frac{(\tau_{K+1})^{n_{K+1}}\prod_{i\in\mathcal{S}^{(1)}
}2^{R_{i}LT}\!\!\!\!}{((\tau_{K+1})^{n_{K+1}}-1)}\left(
\frac{\!\!\int_{\mathbb{R}^{\sum_{i=1}^{K+1}n_{i}}}f(\mathbf{d})d\mathbf{d}}{\prod_{i\in
\{\mathcal{S}^{(1)},K+1\}}V_{f}(\Lambda_{S_{i}})} +
\frac{\int_{\mathbb{R}^{\sum_{i=1}^{K}n_{i}}}f^{\mathcal{S}^{(1)}}(\mathbf{d}_{\mathcal{S}^{(1)}})d\mathbf{d}_{\mathcal{S}^{(1)}}}{(\tau_{K+1})^{n_{K+1}}\prod_{i\in\mathcal{S}^{(1)}
}V_{f}(\Lambda_{S_{i}})}\right),\label{Minkowmarc3}
\end{split}
\end{equation}
which, compared with \eqref{Minkowmarc1}, has an additional term
(second term) in the RHS, (\ref{Minkowmarc3}) where we let
$\mathbf{d}_{\mathcal{S}^{(1)}}=[\mathbf{d}_{i_{1}}^T,...,\mathbf{d}_{i_{|\mathcal{S}^{(1)}|}}^T]^T$,
$i_{1}<\cdots< i_{|\mathcal{S}^{(1)}|}, \forall
i_{j}\in\mathcal{S}^{(1)}$, and the indicator function
$f^{\mathcal{S}^{(1)}}(\mathbf{d}_{\mathcal{S}^{(1)}})=1$ if
$\mathbf{d}_{\mathcal{S}^{(1)}}\in
\mathcal{R}_{\beta}^{{\mathcal{S}^{(1)}}}$, with
$\mathcal{R}_{\beta}^{{\mathcal{S}^{(1)}}}\triangleq
\left\{\mathbf{v}_{{\mathcal{S}^{(1)}}}\in\mathbb{R}^{2|{\mathcal{S}^{(1)}}|M_{u}LT}:\mathbf{v}\in
\mathcal{R}_{\beta},\mathbf{v}_{i}=\mathbf{0},\forall{i}\in
\big\{\{1,...,K+1\}\setminus\mathcal{S}^{(1)}\big\}\right\}$ and the
decision region $\mathcal{R}_{\beta}$ given in Lemma \ref{MHthm1}.
This additional term results in the additional rate constraint
(\ref{ach_modsum2}) compared with Theorem \ref{Thm_ach_one}. Similar
to the derivations of (\ref{Minkowmarc1} a) and \eqref{thm1pf0}, the
LHS of \eqref{Minkowmarc3} equals
\begin{equation}
\frac{1}{|\!\mathcal{E}_{\psi,C^{Lo}}\!|}\!\!\sum_{(\psi^{mod},C_{ur}^{Lo})\in
\mathcal{E}_{\psi,C^{Lo}}}\!\!\!\!\!\!\!\!\!\!\!\!\!\!\left|\mathds{O}_{\mathcal{S}^{(1)}}^{\psi^{mod}_{\Delta}}\cap
\mathcal{R}_{\beta}\right|\!\!=\frac{1}{|\mathcal{C}_{Loe}|}\sum_{C_{ur}^{Lo}\in\mathcal{C}_{Loe}}\left(\frac{1}{|\mathcal{C}_{\psi_{\Delta},\mathrm{E}}|}\sum_{\psi_{\Delta}^{mod}\in\Psi_{\Delta}^{mod}}\sum_{\mathbf{d}\in
\mathds{O}_{\mathcal{S}^{(1)}}^{\psi^{mod}_{\Delta}}}f(\mathbf{d})\right).
\label{eq.bal2}
\end{equation}
Compared with \eqref{thm1pf0}, only the (differential) one-to-one
mapper $\psi_{\Delta}^{one}$ is replaced by $\psi_{\Delta}^{mod}$ in
\eqref{eq.bal2}}. However, unlike  O-MLC in Appendix
\ref{MHthm1proof}, now $\mathcal{C}_{\psi_{\Delta},\mathrm{E}}$ is
not a balanced set, which makes simplifying \eqref{eq.bal2} more
difficult compared with (\ref{thm1pf2} a). To solve this problem, we
need to extend Definition \ref{Def_balance} as follows.
\begin{definition}\emph{(2-partition balanced set):} \label{Def_2_balance}
Following the notation in Definition \ref{Def_balance}, we say that
the set $\mathcal{C}_{\mathrm{E}}$ is $2$-partition balanced if the
non-zero vector set $(\mathcal{C}_{\mathcal{C}_{\mathrm{E}}})^{*}$
can be partitioned as
$(\mathcal{C}_{\mathcal{C}_{\mathrm{E}}})^{*}=\{\mathcal{C}_{\mathcal{C}_{\mathrm{E}},1}^{*},\mathcal{C}_{\mathcal{C}_{\mathrm{E}},2}^{*}\}$,
where every element in
$\mathcal{C}_{\mathcal{C}_{\mathrm{E}},1}^{*}$ is contained in the
same number, denoted by $N_{b,1}$, of $C$ from
$\mathcal{C}_{\mathrm{E}}$ while every element in
$\mathcal{C}_{\mathcal{C}_{\mathrm{E}},2}^{*}$ is also contained in
the same number, denoted by $N_{b,2}$, of $C$ from
$\mathcal{C}_{\mathrm{E}}$.
\end{definition}

For simplifying the RHS of \eqref{eq.bal2}, now we  explore the
properties of the ensemble $\mathcal{C}_{\psi_{\Delta},\mathrm{E}}$
using the 2-partition balanced set in Definition
\ref{Def_2_balance}. Recall that
$\mathcal{C}_{\psi_{\Delta},\mathrm{E}}$  is the ensemble of all
mapper-codebooks (differential) $C_{\psi_{\Delta}^{mod}}$ with
mapper $\psi_{\Delta}^{mod}\in \Psi_{\Delta}^{mod}$, where the
differential super-codewords in $C_{\psi_{\Delta}^{mod}}$ satisfy
the mapping rules of $\psi_{\Delta}^{mod}$. For any user set
$S\subseteq\{1,...,K\}$, let
${\mathcal{C}_{\psi_{\Delta},\mathrm{E}}^{S}}$ be the set of
mapper-codebooks formed by collecting every codebook belonging to
$\mathcal{C}_{\psi_{\Delta},\mathrm{E}}$, but excluding  codewords
$\bar{\mathbf{d}}\notin\mathcal{D}_{S}$  where
$\mathcal{D}_{S}\triangleq\{\bar{\mathbf{d}}:\bar{\mathbf{d}}_{i}\neq\mathbf{0},\forall
i\in S\}$. The fact that for every user set $S$,
${\mathcal{C}_{\psi_{\Delta},\mathrm{E}}^{S}}$ is a 2-partition
balanced set in Definition \ref{Def_2_balance} follows from the
following observations. According to whether the differential
codewords of the relay $\bar{\mathbf{d}}_{r}=\mathbf{0}$ or not, we
can categorize them into two partitions. The differential codewords
in each partition are symmetric according to the proof in Appendix
\ref{MHthm1proof}. Note that $\bar{\mathbf{d}}_{r}=\mathbf{0}$
occurs only in the MS-MLC due to the modulo-sum operation in
Definition \ref{Def_Modulo}. In O-MLC, the one-to-one mapper
guarantees $\bar{\mathbf{d}}_{r}\neq\mathbf{0}$, and results in
simpler \eqref{Minkowmarc1} compared with our target
\eqref{Minkowmarc3}. Now for the first stage, we set
$S=\mathcal{S}^{(1)}$} and the two partitions of
$\mathcal{C}_{\psi_{\Delta},\mathrm{E}}^{\mathcal{S}^{(1)}}$ can be
formed as follows. Let
$\mathcal{C}_{\mathcal{C}_{\psi_{\Delta},\mathrm{E}}^{{\mathcal{S}^{(1)}}},1}^{*}$
and
$\mathcal{C}_{\mathcal{C}_{\psi_{\Delta},\mathrm{E}}^{{\mathcal{S}^{(1)}}},2}^{*}$
be the codeword partitions corresponding to
$\mathcal{C}_{\mathcal{C}_{\mathrm{E}},1}^{*}$ and
$\mathcal{C}_{\mathcal{C}_{\mathrm{E}},2}^{*}$ in Definition
\ref{Def_2_balance} with
$\mathcal{C}_{\mathrm{E}}={\mathcal{C}_{\psi_{\Delta},\mathrm{E}}^{\mathcal{S}^{(1)}}}$
respectively, where
$\mathcal{C}_{\mathcal{C}_{\psi_{\Delta},\mathrm{E}}^{\mathcal{S}^{(1)}},1}^{*}=\{\bar{\mathbf{d}}\in
C_{\psi_{\Delta}^{mod}}^{*}:\bar{\mathbf{d}}_{r}\neq\mathbf{0},\bar{\mathbf{d}}\in\mathcal{D}_{\mathcal{S}^{(1)}}
,\psi_{\Delta}^{mod}\in\Psi_{\Delta}^{mod}\}$, where the codewords
of the relay are distinguishable  since
$\bar{\mathbf{d}}_{r}\neq\mathbf{0}$;
$\mathcal{C}_{\mathcal{C}_{\psi_{\Delta},\mathrm{E}}^{\mathcal{S}^{(1)}},2}^{*}$
is defined similarly but with $\bar{\mathbf{d}}_{r}\neq \mathbf{0}$
replaced by $\bar{\mathbf{d}_{r}}=\mathbf{0}$. Also let the
corresponding balanced numbers of
$\mathcal{C}_{\mathcal{C}_{\mathrm{E}},1}^{*}$ and
$\mathcal{C}_{\mathcal{C}_{\mathrm{E}},2}^{*}$ be $N_{b,1}$ and
$N_{b,2}$ respectively. Now we can simplify the RHS of
\eqref{eq.bal2} using the aforementioned 2-partition balanced set
property and following the proof of the O-MLC counterpart
(\ref{eq.avgmapper2}), while in (\ref{eq.avgmapper2})
$\mathcal{C}_{\psi_{\Delta},\mathrm{E}}$ is a balanced set.
Corresponding to (\ref{eq.avgmapper2}), the term inside the
parentheses on the RHS of \eqref{eq.bal2} now equals
\begin{align}
\sum_{\psi_{\Delta}^{mod}\in\Psi_{\Delta}^{mod}}\sum_{\substack{\mathbf{d}\in
\mathds{O}_{\Delta,\mathcal{S}^{(1)}}^{\psi^{mod}}}}\!\!\!\!\frac{f(\mathbf{d})}{|\mathcal{C}_{\psi_{\Delta},\mathrm{E}}|}
\!\!=\!\!\frac{N_{b,1}}{|\mathcal{C}_{\psi_{\Delta},\mathrm{E}}|}\sum_{{\mathbf{d}}\in
(\Lambda_{C_{ur},1})^{\diamond}
}f(\mathbf{d})+\frac{N_{b,2}}{|\mathcal{C}_{\psi_{\Delta},\mathrm{E}}|}\sum_{{\mathbf{d}}\in
(\Lambda_{C_{ur},2})^{\diamond} }f(\mathbf{d}) \label{basicavg4}
\end{align}
where $(\Lambda_{C_{ur},1})^{\diamond}$ and
$(\Lambda_{C_{ur},2})^{\diamond}$ are the lattice codeword sets
for the 2-partitions corresponding to
$(\Lambda_{C_{ur}})^{\diamond}$ in (\ref{eq.avgmapper2}),
respectively.

Unfortunately, the balanced numbers in \eqref{basicavg4} cannot be
easily computed as in the proof of (\ref{thm1pf2} a) and vary with
${\mathcal{C}_{\psi_{\Delta},\mathrm{E}}^{S}}$ for different sets
$S$. Thus we alternatively show two upper-bounds as
\begin{equation} \label{eq_balance_map}
\frac{N_{b,1}}{|\mathcal{C}_{\psi_{\Delta},\mathrm{E}}|} \leq
\frac{1}{(\tau_{K+1})^{n_{K+1}}-1}, \; \mbox{and} \;\;
\frac{N_{b,2}}{|\mathcal{C}_{\psi_{\Delta},\mathrm{E}}|} \leq
\frac{1}{(\tau_{K+1})^{n_{K+1}}-1},
\end{equation}
where $(\tau_{K+1})^{n_{K+1}}$ is the relay codebook size from
Definition \ref{Def_nested}. Then following similar arguments as
those used in proving (\ref{thm1pf2} a), (\ref{thm1pf2} b) and
(\ref{Minkowmarc1} b) (steps after \eqref{eq.avgmapper2}), we can
prove (\ref{Minkowmarc3}) from \eqref{basicavg4}  and
\eqref{eq_balance_map} with the details omitted. For proving
\eqref{eq_balance_map}, we start with the single user case where
$|\mathcal{S}^{(1)}|=1$ ($\mathcal{S}^{(1)}=\{1,\ldots,K\}=\{1\}$
when  the number of users $K=1$), and then extend to the case
$|\mathcal{S}^{(1)}|=2$ as the upcoming \eqref{2recursive2} and
\eqref{2recursive1}. By repeating this procedure recursively, we
obtain the formulation of balanced numbers in \eqref{eq_balance_map}
as \eqref{recursive4} in the next paragraph and then derive the
upper-bound. When $|\mathcal{S}^{(1)}|=1$,
${\mathcal{C}_{\psi_{\Delta},\mathrm{E}}}$ is a balanced set, and
thus balanced numbers (normalized) are given by
 $\big(N_{b,1}/|\mathcal{C}_{\psi_{\Delta},\mathrm{E}}|\big)_{|\mathcal{S}^{(1)}|=1}=\frac{1}{((\tau_{K+1})^{n_{K+1}}-1)}$
from the  proof of (\ref{thm1pf2} a), and
$\big(N_{b,2}/|\mathcal{C}_{\psi_{\Delta},\mathrm{E}}|\big)_{{|\mathcal{S}^{(1)}|=1}}=0$
by definition. Here the subscript $|\mathcal{S}^{(1)}|=1$ is added
to the notation of the normalized balanced numbers to represent the
upcoming \eqref{2recursive2} and \eqref{2recursive1}. For
$|\mathcal{S}^{(1)}|=2$, the corresponding balanced number for the
partition with $\bar{\mathbf{d}}_{r}=0$ is
\begin{equation} \label{2recursive2}
\left(\frac{N_{b,2}}{|\mathcal{C}_{\psi_{\Delta},\mathrm{E}}|}\right)_{|\mathcal{S}^{(1)}|=2}
=\frac{1}{((\tau_{3})^{n_{3}}-1)}\left(((\tau_{3})^{n_{3}}-1)\left(\frac{N_{b,1}}{|\mathcal{C}_{\psi_{\Delta},\mathrm{E}}|}\right)_{{|\mathcal{S}^{(1)}|=1}}\right).
\end{equation}
To show \eqref{2recursive2}, we  count the occurrence of a
particular super-codeword (differential) $\bar{\mathbf{d}}_{1}
\times \bar{\mathbf{d}}_{2} \times \mathbf{0}$
($\bar{\mathbf{d}}_{r}=\mathbf{0}$) in the overall two-user
mapped-codebook ensemble $\mathcal{C}_{\psi_{\Delta},\mathrm{E}}$,
where user $i$'s codeword (coset leader) is $\bar{\mathbf{d}}_{i},
\;i=1,2$. Let $\psi_{\Delta,i}^{mod}(\bar{\mathbf{d}}_{i})$ be the
(differential) mapper corresponding to user $i$ as in Definition
\ref{Def_Modulo}. First, we compute
$\frac{\big(N_{b,2}\big)_{|\mathcal{S}^{(1)}|=2}}{\big(N_{b,1}\big)_{|\mathcal{S}^{(1)}|=1}}$.
From Definition \ref{Def_balance}, given a particular
$\bar{\mathbf{d}}_{1} \times \bar{\mathbf{d}}_{r,{1}}$ with
$\bar{\mathbf{d}}_{1} \neq \mathbf{0}$ such that
$\bar{\mathbf{d}}_{r,{1}}=\psi_{\Delta,1}^{mod}(\bar{\mathbf{d}}_{1})$,
there will be $\big(N_{b,1}\big)_{|\mathcal{S}^{(1)}|=1}$ possible
mappers $\psi_{\Delta,1}^{mod}$. Now from Definition
\ref{Def_Modulo}, for this partition to have
$\bar{\mathbf{d}}_{r}=\sum_{i=1}^2
\psi_{\Delta,i}^{mod}(\bar{\mathbf{d}}_{i}) =\mathbf{0}$, the
mappers corresponding to user $2$ must satisfy $
(\bar{\mathbf{d}}_{r,{1}}+\psi_{\Delta,2}^{mod}(\bar{\mathbf{d}}_{2}))
\mod \Lambda_{S_{r}}=\mathbf{0} $ since
$\bar{\mathbf{d}}_{r,{1}}=\psi_{\Delta,1}^{mod}(\bar{\mathbf{d}}_{1})$.
Thus for a fixed $\bar{\mathbf{d}}_{r,{1}}$, the vector
$\psi_{\Delta,2}^{mod}(\bar{\mathbf{d}}_{2})$ for the given
$\bar{\mathbf{d}}_{2}$ is also fixed from the definition of the
codomain $C_r^{nest} $ of $\psi_{\Delta,2}^{mod}(\cdot)$ given in
Definition \ref{Def_OneOne}. Also from Definition \ref{Def_setS},
$\bar{\mathbf{d}}_{2} \neq 0$ since the user messages (encoded in
cosets) are with set-$\mathcal{S}^{(1)}$ errors, where
$\mathcal{S}^{(1)}=\{1,2\}$. Then for a fixed
$\bar{\mathbf{d}}_{r,{1}}$, excluding the given
$\bar{\mathbf{d}}_{2}$ and the zero vector, by assigning the mapping
rules for the remaining  $(\tau_{2})^{n_{2}}-2$ points in the domain
of $\psi_{\Delta,2}^{mod}(\cdot)$, there are
$\prod_{i=2}^{(\tau_{2})^{n_{2}}-1}((\tau_{3})^{n_{3}}-i)$ possible
injective mappers $\psi_{\Delta,2}^{mod}$ where $(\tau_{i})^{n_{i}}$
is transmitter $i$'s (users and relay) codebook size. Note that to
make
$(\psi_{\Delta,2}^{mod}(\bar{\mathbf{d}}_{2})+\bar{\mathbf{d}}_{r,{1}})
\mod \Lambda_{S_{r}}=\mathbf{0}$, it is required that
$\bar{\mathbf{d}}_{r,{1}}\neq \mathbf{0}$ since
$\bar{\mathbf{d}}_{2} \neq 0$. As there are a total of
$((\tau_{3})^{n_{3}}-1)$ possible $\bar{\mathbf{d}}_{r,{1}}\neq
\mathbf{0}$ in the relay's (differential) codebook, we have
$\frac{\big(N_{b,2}\big)_{|\mathcal{S}^{(1)}|=2}}{\big(N_{b,1}\big)_{|\mathcal{S}^{(1)}|=1}}=((\tau_{3})^{n_{3}}-1)\prod_{i=2}^{(\tau_{2})^{n_{2}}-1}((\tau_{3})^{n_{3}}-i)$.
Also
$\frac{|\mathcal{C}_{\psi_{\Delta},\mathrm{E}}|_{{|\mathcal{S}^{(1)}|=2}}}{|\mathcal{C}_{\psi_{\Delta},\mathrm{E}}|_{{|\mathcal{S}^{(1)}|=1}}}=((\tau_{3})^{n_{3}}-1)\prod_{i=2}^{(\tau_{2})^{n_{2}}-1}((\tau_{3})^{n_{3}}-i)$
by counting all possible injective mappers $\psi_{\Delta,2}^{mod}$
of user $2$. Thus \eqref{2recursive2} is valid. Similar to
\eqref{2recursive2}, for $|\mathcal{S}^{(1)}|=2$, the corresponding
balanced number for the partition with $\bar{\mathbf{d}}_{r} \neq 0$
is
\begin{align}
\left(\frac{N_{b,1}}{|\mathcal{C}_{\psi_{\Delta},\mathrm{E}}|}\right)_{{|\mathcal{S}^{(1)}|=2}}&=\frac{1}{((\tau_{3})^{n_{3}}-1)}\left(((\tau_{3})^{n_{3}}-2)\left(\frac{N_{b,1}}{|\mathcal{C}_{\psi_{\Delta},\mathrm{E}}|}\right)_{{|\mathcal{S}^{(1)}|=1}}+1\cdot\left(\frac{N_{b,2}}{|\mathcal{C}_{\psi_{\Delta},\mathrm{E}}|}\right)_{{|\mathcal{S}^{(1)}|=1}}\right).
\label{2recursive1}
\end{align}
The proof of \eqref{2recursive1} is similar to that for
\eqref{2recursive2}, but now
$(\psi_{\Delta,2}^{mod}(\bar{\mathbf{d}}_{2})+\bar{\mathbf{d}}_{r,{1}})
\mod \Lambda_{S_{r}} \neq \mathbf{0}$. The first term in the
parenthesis on the RHS of \eqref{2recursive1} corresponds to the
case $\bar{\mathbf{d}}_{r,{1}} \neq 0$ while the second term
corresponds to the case $\bar{\mathbf{d}}_{r,{1}} = 0$. The details
are omitted.

Finally, by repeating the arguments in the previous paragraph we can
find the balanced numbers for $|\mathcal{S}^{(1)}|=3$ with
\eqref{2recursive2} and \eqref{2recursive1}, and so on. Then for the
balanced numbers when $|\mathcal{S}^{(1)}|=K$ , we have
\begin{align}
\frac{N_{b,1}}{|\mathcal{C}_{\psi_{\Delta},\mathrm{E}}|}&=\frac{1}{(\tau_{K+1})^{n_{K+1}}((\tau_{K+1})^{n_{K+1}}-1)^{|\mathcal{S}^{(1)}|}}\left(((\tau_{K+1})^{n_{K+1}}-1)^{|\mathcal{S}^{(1)}|}+(-1)^{|\mathcal{S}^{(1)}|+1}\right) \notag \\
\frac{N_{b,2}}{|\mathcal{C}_{\psi_{\Delta},\mathrm{E}}|}&=\frac{1}{(\tau_{K+1})^{n_{K+1}}((\tau_{K+1})^{n_{K+1}}-1)^{|\mathcal{S}^{(1)}|}}\left(((\tau_{K+1})^{n_{K+1}}-1)^{|\mathcal{S}^{(1)}|}+((\tau_{K+1})^{n_{K+1}}-1)(-1)^{|\mathcal{S}^{(1)}|}\right)
\label{recursive4},
\end{align}
where $(\tau_{K+1})^{n_{K+1}}$ is the relay codebook size. On noting
that $\frac{1}{((\tau_{K+1})^{n_{K+1}}-1)^{|\mathcal{S}^{(1)}|}}\leq
\frac{1}{((\tau_{K+1})^{n_{K+1}}-1)}$ for
$|\mathcal{S}^{(1)}|\geq1$, together with (\ref{recursive4}), one
can show that \eqref{eq_balance_map} is valid. Then our proof for
\eqref{Minkowmarc3} is complete.

\subsection{Proof  of   (\ref{Minkowmarc1} b) and \eqref{Minkowmarc2} } \label{appendixderiveof16} We can rewrite \eqref{thm1pf2} as
\begin{align}
&\frac{1}{((\tau_{K+1})^{n_{K+1}}-1)|\mathcal{C}_{Loe}|}\sum_{C^{Lo}_{ur}\in\mathcal{C}_{Loe}}\sum_{\mathbf{d}\in(\Lambda_{{C}_{ur}})^{\star}}f(\mathbf{d})
\\
=&\frac{1}{((\tau_{K+1})^{n_{K+1}}-1)}\sum_{\mathbf{z}\in\left(\mathbb{Z}^{n}\right)^{\star}
:
\bar{\mathbf{z}}_{\underline{p}}=\mathbf{0}}f(\underline{\gamma}\mathbf{z})+\frac{1}{|\mathcal{C}_{Loe}|}\sum_{C^{Lo}_{ur}\in\mathcal{C}_{Loe}}\sum_{\mathbf{a}\in{(C^{Lo}_{ur})^{*}}}\left[\sum_{\mathbf{z}\in\left(\mathbb{Z}^{n}\right)^{\star}
:
\bar{\mathbf{z}}_{\underline{p}}=\mathbf{a}}f(\underline{\gamma}\mathbf{z})\right],
\label{thm1pf3}
\end{align}
 where   $\underline{\gamma}\mathbf{z}$ is defined in
(T1.23) in Table \ref{notations_table}. In (\ref{thm1pf3}),  we
define $(\mathbb{Z}^{n})^{\star}\triangleq\left\{\mathbf{z}\in
\mathbb{Z}^{n}:\mathbf{z}_{i}\neq\mathbf{0},  \forall
i\in\{1,...,K+1\}\right\} $ and $\bar{\mathbf{z}}_{\underline{p}}$
is formed by applying modulo $p_{i}$ operation  on  elements of
$\mathbf{z}_{i}$ ((T1.22) in Table \ref{notations_table}). Now for
summation in the second term of (\ref{thm1pf3}), we separate the
summation over $\mathbf{a}\in ({C^{Lo}_{ur}})^{*}$ by the cases $\{
\mathbf{a}_{u}\neq\mathbf{0},\mathbf{a}_{r}=\mathbf{0}\}$, $\{
 \mathbf{a}_{r}\neq\mathbf{0},\mathbf{a}_{u}=\mathbf{0}\}$ and $\{
\mathbf{a}_{r}\neq\mathbf{0}, \mathbf{a}_{u}\neq\mathbf{0} \}$. By
averaging over $\mathcal{C}_{Loe}$ for these three cases, we have
(\ref{thm1pf5}), (\ref{thm1pf6}) and (\ref{thm1pf7}), respectively,
\allowdisplaybreaks
\begin{align}
&\frac{1}{|\mathcal{C}_{Loe}|}\sum_{C^{Lo}_{ur}\in\mathcal{C}_{Loe}}\sum_{\mathbf{a}\in{(C^{Lo}_{ur})^{*}}}\left[\sum_{\mathbf{z}\in\left(\mathbb{Z}^{n}\right)^{\star}
:
\bar{\mathbf{z}}_{\underline{p}}=\mathbf{a}}f(\underline{\gamma}\mathbf{z})\right] \notag\\
=&\sum_{S\subseteq\{1,...,K\},S\neq\phi}\left[\frac{\prod_{i\in
S}(p_{i}^{k_{i}}-1)}{\prod_{i\in
S}(p_{i}^{n_{i}}-1)}\cdot\sum_{\mathbf{z}\in
\left(\mathbb{Z}^{n}\right)^{\star}:
(\bar{\mathbf{z}}_{\underline{p}})_{i}\neq \mathbf{0},  i\in
{S},(\bar{\mathbf{z}}_{\underline{p}})_{i'}=\mathbf{0},i' \in
{\{S^{c},K+1\}}}f(\underline{\gamma}\mathbf{z})\right]
\label{thm1pf5}
\\
+&\left[\frac{(p_{i}^{k_{K+1}}-1)}{(p_{i}^{n_{K+1}}-1)}\cdot\sum_{\mathbf{z}\in
\left(\mathbb{Z}^{n}\right)^{\star}:
(\bar{\mathbf{z}}_{\underline{p}})_{K+1}\neq \mathbf{0},
(\bar{\mathbf{z}}_{\underline{p}})_{i'}=\mathbf{0},i' \in
\mathcal{S}^{(1)}}f(\underline{\gamma}\mathbf{z})\right]
\label{thm1pf6}
\\
+&\sum_{S\subseteq\{1,...,K\},S\neq\phi}\left[\frac{\prod_{i\in
\{S,K+1\}}(p_{i}^{k_{i}}-1)}{\prod_{i\in
\{S,K+1\}}(p_{i}^{n_{i}}-1)}\cdot\sum_{\mathbf{z}\in
\left(\mathbb{Z}^{n}\right)^{\star}:
(\bar{\mathbf{z}}_{\underline{p}})_{i}\neq \mathbf{0},  i\in
{\{S,K+1\}},(\bar{\mathbf{z}}_{\underline{p}})_{i'}=\mathbf{0},i'
\in {S^{c}}}f(\underline{\gamma}\mathbf{z})\right] \label{thm1pf7}
\\
\rightarrow & \frac{1}{((\tau_{K+1})^{n_{K+1}}-1)\prod_{i\in
\{\mathcal{S}^{(1)},K+1\}}V_{f}(\Lambda_{C_{i}})}\int_{\mathbb{R}^{n}}f(\mathbf{d})d\mathbf{d}
\label{thm1pf8}
\end{align}
as $p_{i}\rightarrow\infty$, $\gamma_{i}\rightarrow 0$ (Definition
\ref{Def_logier}). Since $f$ has a bounded support ($f$ vanishes at
infinity), with the definition of
$\left(\mathbb{Z}^{n}\right)^{\star}$, the first term in
(\ref{thm1pf3}) vanishes for sufficiently large
$\gamma_{i}p_{i}\rightarrow\infty$ as shown in
\cite{latticemac},\cite{avgboundlattice}.  The terms in
(\ref{thm1pf5}) also vanish  by noting that at least one of elements
of $\mathbf{z}_{K+1}$  is equal to the multiples of $p_{K+1}$, which
results in $f(\underline{\gamma}\mathbf{z})\rightarrow 0$ in
(\ref{thm1pf5}). The term in (\ref{thm1pf6}) follows similarly.
Finally, the term in (\ref{thm1pf7}) approaches to (\ref{thm1pf8})
for $S=\mathcal{S}^{(1)}$, and vanish otherwise in a way similar to
(\ref{thm1pf5}), (\ref{thm1pf6}), as $\gamma_{i}\rightarrow0$ with
$p_{i}^{n_{i}-k_{i}}\gamma_{i}^{n_{i}}=V_{f}(\Lambda_{C_{i}})$ fixed
as in those \cite{latticemac}. From Definition \ref{Def_nested},
$V_{f}(\Lambda_{S_{i}})/V_{f}(\Lambda_{C_{i}})=2^{R_iLT}=(\tau_{i})^{n_{i}}$,
then (\ref{Minkowmarc1} b) can be obtained from \eqref{thm1pf8}.
Finally, (\ref{Minkowmarc2}) can be obtained from (\ref{Minkowmarc1}
b) by following the footsteps in \cite{latc}.

\subsection{Proof of
(\ref{eq.DMTanalsis1})} \label{appendixthmdmtone}
\begin{proof}
 For the $K$ users, we use the
self-similar nested lattice (Definition \ref{Def_nested}) where
$\Lambda_{S_{i}}=\tau_{i}\Lambda_{C_{i}}$, $\tau_{i}=\lfloor
\rho_{d}^{\frac{r_{i}}{2M_{u}}}\rfloor$ in order to satisfy the
transmission rate constraint $R_{i}(\rho_{d})\doteq
r_{i}\log\rho_{d}$. The ensemble $\mathcal{E}_{\psi,C^{Lo}}$ defined
in the proof of Theorem \ref{Thm_ach_one} with $k_{i}=1$ (Definition
\ref{Def_logier}) is considered, on which the corresponding lattices
ensemble  is then expurgated in a way similar to that in the proof
of Theorem 6 in \cite{latc}.
 We denote the expurgated
 ensemble of codebooks, $C_{code}$ (i.e.,  $C_{\psi^{one}}$ given the corresponding lattices in the expurgated lattices ensemble), as $\mathcal{C}_{code}^{exp}$. Then the
average error  probability, $P_{e}(\rho_{d})$ in (T1.20) of Table
\ref{notations_table}, can be upper bounded by
\begin{align}
P_{e}(\rho_{d})  \triangleq
E_{\mathcal{C}_{code}^{exp},\mathbf{H}_{dst}}\left[Pr(Er|C_{code},\mathbf{H}_{dst})\right]
 \leq
Pr(\mathcal{O})+E_{\mathcal{C}_{code}^{exp},\mathbf{H}_{dst}}\left[Pr(Er,\mathcal{O}^{c}|C_{code},\mathbf{H}_{dst})\right]
\label{eq.1}
\end{align}
where $Pr(Er|C_{code},\mathbf{H}_{dst})$ is the probability of the
event that given a $\{C_{code},\mathbf{H}_{dst}\}$, not all users
are correctly decoded at the destination and $\mathcal{O}$ denotes
for the outage event set of $\mathbf{H}_{dst}$ ($\mathbf{H}_{dst}$
 does not satisfy (\ref{ach_one})).
For the second term on the RHS of the inequality in (\ref{eq.1}), by
averaging the
 term $Pr(Er,\mathcal{O}^{c}|{C}_{code},\mathbf{H}_{dst})$ over
 $C_{code}\in\mathcal{C}_{code}^{exp}$ and then over $\mathbf{H}_{dst}\in
 \mathcal{O}^{c}$,
 we will show
\begin{equation}
E_{\mathbf{H}_{dst}}\left[Pr(Er,\mathcal{O}^{c}|\mathbf{H}_{dst})\right]
\dot{=}Pr(\mathcal{O}) 
\end{equation}
Following the  steps similar to those in \cite{latticemac} and
\cite{latc}, considering a tuple of multiplexing gains, $r_{i}$, to
meet a diversity requirement $d$ for each user as in \cite{DMT_MAC},
 given a $\mathbf{H}_{dst}$,
 we have,
\begin{equation}
\begin{split}
Pr(Er, \mathcal{O}^{c}|\mathbf{H}_{dst})\dot{\leq}
&(1+\delta)\frac{\tau_{K+1}^{n_{K+1}}}{\tau_{K+1}^{n_{K+1}}-1}\left(
\sum_{S\subseteq\{1,...,K\},S\neq\phi}\rho_{d}^{LT\sum_{i\in
S}r_{i}}\bigg(\frac{KM_{u}+M_{r}}{|S|M_{u}+M_{r}}\bigg)^{(|S|M_{u}+M_{r})LT}\right.\\\
&
\;\;\;\;\;\;\;\;\;\;\;\;\;\;\;\;\;\;\;\;\;\;\;\;\;\;\;\;\;\;\;\;\;\;\;\;\;\;\;\;\;\;\;\;\;\;\;\;\;\;\;\;\;\left.\cdot\det\left(\mathbf{I}_{2(|S|M_{u}+M_{r})LT}+\left(\mathbf{H}_{dst}^{\{S,K+1\}}\right)^{H}\mathbf{H}_{dst}^{\{S,K+1\}}\right)^{-1}\right),\label{eq.3}
\end{split}
\end{equation}
where $p_{i}\rightarrow\infty, \forall i$ and  $\delta>0$.

 Let
$Pr(\mathcal{O})\doteq\rho_{d}^{-d(\mathbf{r})}$.  Although $Pr(O)$
does not necessarily  guarantee the minimum outage probability, it
suffices for the DMT analysis as indicated in \cite{diva}. The
explicit formulation of
 $d(\mathbf{r})$
 is
generally difficult to obtain since the joint probability density
function (pdf) of eigenvalues of
$(\mathbf{H}_{dst}^{\{S,K+1\}})^{H}\mathbf{H}_{dst}^{\{S,K+1\}}$ is
generally not easy to evaluate. However, from Theorem 3.2.17 of
\cite{dmtbook}, it can be seen that the joint pdf of these
eigenvalues is a continuous function. Therefore, by choosing  a
sufficiently large, but finite $T$, such that  the term on the RHS
in (\ref{eq.3}) decays  fast enough, we can prove that
$E_{\mathbf{H}_{dst}}[Pr(Er,\mathcal{O}^{c}|\mathbf{H}_{dst})]$ is
exponentially equal to $Pr(O)$ using the techniques similar to those
in
 \cite{latticemac}, \cite{latc},\cite{DMT_MAC} and \cite{diva}.  Together with
(\ref{eq.1}), we obtain (\ref{eq.DMTanalsis1}). Note the rate loss
terms in (\ref{eq.3}) are exponentially negligible  (independent of
$\rho_{d}$) in the DMT analysis.
\end{proof}

\vspace{-5mm}

\renewcommand{\baselinestretch}{1}

\bibliographystyle{IEEEtran}
\bibliography{IEEEabrv,reference}

%
%
%
%
%
%


\end{document}